\documentclass[12pt]{amsart}
\usepackage{amssymb,amscd,amsthm,extarrows,verbatim,amsmath,color,fancyhdr,mathrsfs}
\usepackage{algorithm}
\usepackage{algpseudocode}
\usepackage{tikz-cd}
\usepackage{multirow}
\usepackage{amsaddr}
\usepackage{graphicx}
\usepackage{turnstile}
\usepackage{arydshln}
\usepackage{enumitem}
\usepackage{appendix}
\usepackage{nicefrac}
\usepackage{xcolor}
\usepackage{pagecolor}

\usepackage[letterpaper, left=2.5cm, right=2.5cm, top=2.5cm,
bottom=2.5cm,dvips]{geometry}

\newcommand{\E}{\mathbf{E}}
\newcommand{\R}{\mathbf{R}}

\newcommand{\p}{\mathbf{P}}
\newcommand{\sgn}{\operatorname{sgn}}
\def\arrvline{\hfil\kern\arraycolsep\vline\kern-\arraycolsep\hfilneg}

\newtheorem{thm}{Theorem}[section]
\newtheorem{prop}[thm]{Proposition}
\newtheorem{lem}[thm]{Lemma}

\newtheorem{cor}[thm]{Corollary}

\newtheorem{defn}[thm]{Definition}
\title[Partitioned $K$-nearest neighbor local depth]
{Partitioned $K$-nearest neighbor local depth
\\for scalable comparison-based learning}
\author{Jacob D.\ Baron \hspace{.6cm} R.W.R.\ Darling \hspace{.6cm} J. Laylon Davis \hspace{.6cm} R.\ Pettit}
\address{National Security Agency, Fort George G.\ Meade, MD 20755-6844, USA}
\date{\today}

\begin{document}
\maketitle


\begin{abstract}
A triplet comparison oracle on a set $S$ takes an object $x \in S$
and for any pair $\{y, z\} \subset S \setminus \{x\}$ declares which of $y$ and $z$ is more similar to $x$. 
Partitioned Local Depth (\texttt{PaLD}) supplies a principled non-parametric partitioning 
of $S$ under such triplet comparisons
but needs $O(n^2 \log{n})$ oracle calls and $O(n^3)$ post-processing steps.

We introduce Partitioned Nearest Neighbors Local Depth (\texttt{PaNNLD}),
a computationally tractable variant of \texttt{PaLD} leveraging the $K$-nearest neighbors digraph
on $S$. \texttt{PaNNLD} needs only $O(n K \log{n})$ oracle calls,
by replacing an oracle call by a coin flip when 
neither $y$ nor $z$ is adjacent to $x$ in the undirected version of the $K$-nearest neighbors digraph.
By averaging over randomizations, \texttt{PaNNLD} subsequently
requires (at best) only $O(n K^2)$ post-processing steps.
Concentration of measure shows
that the probability of randomization-induced error $\delta$ in \texttt{PaNNLD}
is no more than $2 e^{-\delta^2 K^2}$.
\end{abstract}

{\small
\textbf{MSC class:} Primary: 90C35; Secondary: 06A07
}

\vspace{.4in}
\begin{center}
    \emph{I have always depended on the kindness of strangers.}
\end{center}
\vspace{.15in}
\hspace{2.5in} Tennessee Williams \textit{A Streetcar Named Desire}, 1947
\vspace{.2in}

\setcounter{tocdepth}{1}
\tableofcontents

\section{Introduction}
\subsection{Triplet comparisons}\label{s:triplets}
Unsupervised learning seeks to partition a collection $S$ of objects in
such a way that objects in the same subset are similar to each other.
Typical algorithms represent objects as vectors, or
points in a metric space, or vertices in a graph. See \cite{emr} for a wide-ranging survey.
\textit{Comparison-based learning} needs no such embedding, nor any numerical measure of similarity.
\textit{Triplet comparison} suffices: for any $x,y,z \in S$, declare which of $y$ and $z$ is more similar to $x$.

If objects in $S$ are of some type $T$, our data structure is a
\texttt{Map} (i.e. a set of key-value pairs), whose keys are the elements of $S$, and whose values
are \texttt{Comparator}s \cite{javacomp} of objects of type $T$. The \texttt{compare}$(x; y, z)$ method
of the \texttt{Comparator} associated with $x \in S$ 
reports which of $y, z \in S$ is more similar to $x$:
it returns an integer which is negative if 
$y$ is more similar to $x$ than $z$ is, positive if vice versa, and zero
in the case of a tie. We assume:
\begin{enumerate}
    \item \textit{Antisymmetry: } $\sgn($\texttt{compare}$(x; z, y)) = - \sgn($\texttt{compare}$(x; y, z))$.
    \item \textit{Transitivity: }\texttt{compare}$(x; y, z) > 0$ \texttt{compare}$(x; z, w) > 0$
    imply \texttt{compare}$(x; y, w) > 0$.
    \item \textit{Totality\footnote{So we disallow ties (i.e.\
    distinct $y$ and $z$ that are equally similar to $x$), although \cite{ber} allows ties.
    }: }\texttt{compare}$(x; y, z) = 0$ if and only if $y = z$.
    \item \textit{Autosimilarity: }\texttt{compare}$(x; x, y) < 0$ for all $y \neq x$.
\end{enumerate}
Conditions (1) and (2) are standard for \texttt{Comparator} \cite{javacomp}, (3) is optional, and (4) is specific to this context.
Condition (4) says that no $y \neq x$ is as similar to $x$ as $x$ is to itself.

In our triplet comparisons, no particular relationship is assumed between \texttt{compare}$(x; y, z)$,
\texttt{compare}$(y; z, x)$, and \texttt{compare}$(z; x, y)$. Comparison need not be based on a symmetric 
distance or similarity function. It is, however, interpretable in terms of graph orientations (Section \ref{s:graphorientation}).
Throughout the symbol $n$ will denote the (finite) size of $S$.

Algorithms already implemented in the context of a \texttt{Comparator} include:
\begin{itemize}
    \item
    efficient $K$-nearest neighbor descent \cite{don, bar, hag};
    \item
    random forest classifiers \cite{hagrf};
    \item
    active learning, \cite{kan}, and supervised learning based on samples \cite{per};
      \item
    recovering a hierarchical clustering \cite{ema};
    \item
    finding a clustering \cite{bia} or recovering a planted clustering \cite{per2};
    \item
    exploratory data analysis via local depth \cite{ber, darp}.
\end{itemize}

\subsection{Data science examples of triplet comparisons}
Here are some examples where triplet comparison appears more natural than vector or metric embedding.
In many cases $n$ could be in the tens of millions.

    \subsubsection{Weighted categorical features: } \label{s:bipartite}
    A bipartite graph encodes the appearance of features $F$ (right vertices) among a set $S$ of objects (left vertices).
    An edge $\{x, f\}$ of weight $w$ means that feature $f \in F$ has strength $w$ at object $x$.
    The value of $\texttt{compare}(x;y,z)$ depends on total strength at $x$ of features shared with $x$ (by $y$ and $z$).
    In topic modelling, $S$ consists of documents, and $F$ consists of keywords.
    
\subsubsection{Lexicographic-order comparator for multi-type features: }
    Each object in $S$ is a collection of data fields associated with an entertainment product, such as a movie
    or a song track. For a movie, these fields include title, director, producer, principal actors, date, etc.
    Suppose these fields $F_1, F_2, \ldots, F_r$ are listed in order of importance, and within field
    $F_j$ there is a \texttt{Comparator} $C_j$ at $x$, for each $x$. 
    To decide which of movies $y$ and $z$ is more similar to $x$,
    look at the field $F_1$, and answer ``$y$'' if \texttt{Comparator} $C_1$ at $x$ ranks $y$ before $z$, or ``$z$'' if vice-versa. If $C_1$ returns a tie, we look at the field $F_2$, and apply \texttt{Comparator} $C_2$, and so forth. The
    Java\texttrademark Language \cite{javacomp} calls this a \textit{lexicographic-order comparator}.
    
 \subsubsection{Geostatistics -- inhomogeneous point clouds: } \label{s:geostats}
    Each object in $S$ is a geographic location, i.e.\ a latitude-longitude pair of an event, or 
    fixed object, related to human activity\footnote{
    Typical data source: \texttt{https://location.services.mozilla.com/}}. Metric embedding is possible here, but if the locations are drawn from a mix of densely and sparsely populated
    environments, absolute distance is a poor measure of closeness. (Compare a mile in Manhattan to a mile in the Sahara.) The non-parametric approach 
    uses a \texttt{Comparator} at point $x$ which ranks other points according to
    travel time or road distance from $x$. 
    
\subsubsection{Genomics -- matching algorithms: }
    Each object in $S$ is a genomic sequence for a specific organism. These organisms may come from
    various species, whose genomic sequences have different lengths. Similarity between two sequences
    may be measured, as in \cite{doe}, by the maximum cardinality set of non-conflicting conserved adjacencies. Clustering such sequences may be
    a useful tool in setting species boundaries.

\subsection{Other sorts of comparison}
Other authors, such as \cite{add, bia, ema, emr, hag, per}, assume dissimilarity is measured by a metric on $S$. Symmetry and triangle inequality assumptions are unnecessary for partitioned local depth, and are violated in some of the examples
above. We return to this topic below when we present
triplet comparison in terms of graph orientations.

Besides triplet comparison, some authors use quadruplet comparison \cite{add, per}. The implicit assumption here
is that a uniform scale of distance is present throughout the data set, which is false for large-scale
human spatial data (Section \ref{s:geostats}).

There is an extensive literature \cite{aga, and} on \textit{ordinal embedding} of $S$ into
some Euclidean space, so that triplet and quadruplet comparisons are compatible with
distance comparisons. Recently a scalable algorithm was proposed by Anderton and Aslam \cite{and}.
However ordinal embedding is not a pre-requisite to 
\texttt{PaLD} or the present work.

\subsection{Outline of the paper}
Berenhaut, Moore and Melvin \cite{ber} proposed partitioned local depth (\texttt{PaLD}) as a non-parametric approach
to unsupervised learning based on triplet comparisons.

Our objective is to leverage the $K$-nearest neighbors digraph
on $S$ to design a computationally tractable variant of \texttt{PaLD}, whose approximation error can be 
estimated. We produce
a general comparison-based unsupervised learning algorithm, \texttt{PaNNLD}, applicable to data sets
of millions of objects. Practical implementation \cite{darp} will be described further in a sequel.

Section \ref{s:formalism} offers a formalism for triplet comparison. 
Section \ref{s:pald} reviews partitioned local depth (\texttt{PaLD}) from \cite{ber} with some
additional results. Section \ref{s:pannld} defines the new \texttt{PaNNLD} algorithm via
 a method called \textit{stranger randomization}, which gives consistent but unrecorded random answers to
triplet comparisons not involving near neighbors. 
In Section \ref{s:pannld} three theorems are stated whose proofs occupy the remainder of the paper;
two relate to algorithmic implementation of \texttt{PaNNLD}, and the third to 
an exponential bound on the error introduced by stranger randomization, using
concentration of measure \cite{bou}.

\section{Triplet comparisons, partial orders, and orientations of the line graph}\label{s:formalism}
\subsection{Partial order formulation}\label{s:totalorder}
In Section \ref{s:triplets} we introduced triplet comparison in terms of the 
\texttt{Comparator} interface \cite{javacomp} from computer science, which is the workhorse of
our Java\texttrademark  implementation \cite{darp}.
The associated order-theoretic notion is a family of \textit{total orders} $\preceq_x$ on $S$,
for each $x \in S$, with the interpretation:
$y \preceq_x z$ (or $z \succeq_x y$) means $y$ is at least as similar to $x$ as $z$ is, i.e.\ 
    \texttt{compare}$(x; y, z) \leq 0$. 
    Augment the defining properties of a total order, namely antisymmetry, transitivity, and totality, by
    a further property of autosimilarity.
\begin{enumerate}
     \item[(I)] \textit{Antisymmetry:} $y \preceq_x z \preceq_x y$ implies $y=z$.
     \item[(II)] \textit{Transitivity:} $y \preceq_x y' \preceq_x y''$ implies $y \preceq_x y''$.
     \item[(III)] \textit{Totality:} $y \preceq_x z$ or $z \preceq_x y$ for every $y, z$.
     \item[(IV)]  \textit{Autosimilarity:} $x \preceq_x y$ for all $y \in S$.
\end{enumerate}

We shall use the notation $y \prec_x z$ to mean $y \preceq_x z$ and $y \neq z$.

Baron and Darling \cite{bar} call such a family of total orders $\{\preceq_x, x \in S\}$ a \textbf{ranking system}.
There are $\binom{n-1}{2}$ nontrivial\footnote{Exclude relations taking the form $x \preceq_x y$ or $y \preceq_x y$.} 
ordered pairs in the binary relation $\preceq_x$, for each $x$. 
Restricted to these nontrivial pairs, we can view $\preceq_x$ as a total order $\preceq_x'$ on the set $\{\{x,y\}\::\:y \in S \setminus \{x\}\}$
of \emph{pairs} of elements of $S$ containing $x$, in the following way: 
$y \preceq_x z \text{ iff } \{x,y\} \preceq_x' \{x,z\}$.
(So the elements of the relation $\preceq_x'$ are ordered pairs of unordered pairs of elements of $S$, like $(\{x,y\},\{x,z\})$.) When the union of these $n$ relations $\preceq_x'$ is a subset of some partial order $\preceq^\star$ on the set $\binom{S}{2}$ of all unordered pairs of elements of $S$,
then the ranking system is defined in \cite{bar} to be \textbf{concordant}. 

A symmetric dissimilarity function $\sigma: S \times S \to \R$ (such as a metric on $S$), without ties,
induces the ranking system
in which $y \prec_x z$ whenever $\sigma(x, y) < \sigma(x, z)$. 
Such a ranking system is necessarily concordant, because the linear order of inter-point distances in $\mathbf{R}$ provides a valid $\preceq^\star$.
Baron and Darling prove in \cite[Lemma 5.5]{bar} the converse, namely that all concordant ranking systems 
arise from a metric space in this way.

A ranking system is concordant if and only if it admits no ``cycles'' of relations such as $y \prec_x z$, $z \prec_y x$, and $x \prec_z y$ (not necessarily of length 3). Most ranking systems are non-concordant; see \cite[Sec.\ 5]{bar}.

\begin{figure}
\caption{\textit{Complete graph $K_5$ on a 5-element set, and the line graph
of $K_5$ whose nodes are the edges of $K_5$. A random acyclic orientation has been assigned to
the line graph. Interpret $\{x, y\} \longrightarrow \{x, z\}$ to mean $y \prec_x z$.
Ignoring arcs that follow directed paths of two or more arcs,
this oriented line graph is equivalent to
the \emph{Hasse diagram} of a partial order $\preceq^\star$ for the corresponding concordant ranking system; cf.\ Remark after Proposition \ref{p:3notions}.
}
}
\label{f:linek5}
\begin{center}
\scalebox{0.4}{\includegraphics{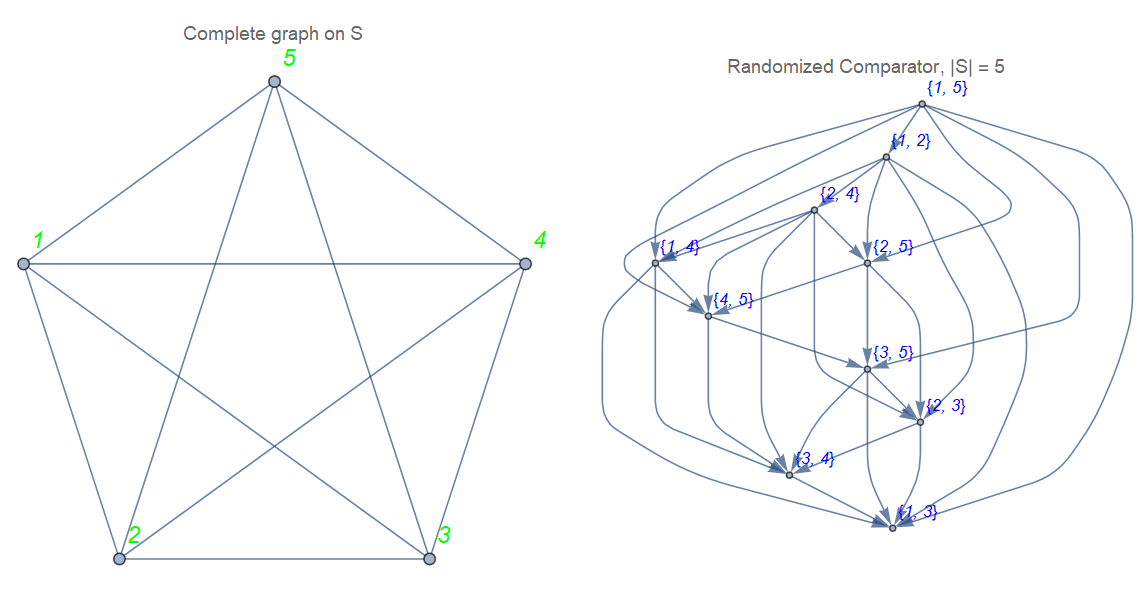} }
\end{center}
\end{figure}

\subsection{Graph orientation language for triplet comparison}\label{s:graphorientation}

Consider the complete graph $\mathcal{K}_S$ on the vertex set $S$. Its edges are the nodes of
the \textit{line graph} $L(\mathcal{K}_S)$, where two edges are treated as adjacent if they
have a vertex in common.
See Figure \ref{f:linek5} for an example.
Summing vertex degrees, and dividing by two, gives an edge count of $n(n-1)(n-2)/2$.

We can recast Section \ref{s:totalorder} in the language of graph orientations. To make
autosimilarity explicit, augment\footnote{
In the sequel,  the original line graph is sometimes preferable, keeping the autosimilarity condition implicit.
} the line graph $L(\mathcal{K}_S)$ by adding $n$ extra vertices $\{\{x, x\},  x \in S\}$ and
declaring that $\{x, x\}$ is adjacent to $\{x, y\}$ for all $y \in S \setminus \{x\}$.
Denote the augmented line graph by $L_+(\mathcal{K}_S)$. 
The latter has $\binom{n+1}{2}$ vertices and $n^2 (n-1)/2$ edges.

For each $x \in S$, the subgraph of $L_+(\mathcal{K}_S)$ induced by the $n$-vertex subset
\[
V_x:=\{ \{x, y\}, y \in S\}
\]
is the complete graph on $V_x$, which we denote $\mathcal{K}_{V_x}$.
An \textbf{orientation} of an undirected graph means an assignment of a direction to every edge of the graph,
and a \textbf{tournament} is a complete graph with an orientation.
A ranking system on $S$ induces a tournament on the vertex set $V_x$ for each $x$, where
the directed edge
\begin{align}\label{e:direction}
    \{x,y\} \longrightarrow \{x,z\}, \; y \neq z \hspace{.5in} \text{corresponds to} \hspace{.5in} y \prec_x z,
\end{align}
and to \texttt{compare}$(x; y, z) < 0$.  The vertex $\{x, x\}$ is a source, meaning a vertex with no
incoming directed edges.
Since $\preceq_x$ is a total order on $S$, this tournament on $V_x$ is transitive (obvious translation
of condition (II)). The total order on $V_x$ makes the tournament acyclic.

\begin{prop}\label{p:3notions}
Notions (a), (b), and (c) are equivalent:
\begin{enumerate}
    \item[(a)] A \texttt{Comparator} at each $x \in S$, satisfying (1)--(4) of Section \ref{s:triplets}.
    \item[(b)] A ranking system $\{\preceq_x, x \in S\}$, satisfying (I)--(IV)  of Section \ref{s:totalorder}.
    \item[(c)] An orientation of the augmented line graph $L_+(\mathcal{K}_S)$, so that for each $x$,
    the induced tournament on ${V_x}$ is acyclic, with a source at $\{x, x\}$.
\end{enumerate}
\end{prop}

\textbf{Remark:} 
Part (c) of Proposition \ref{p:3notions} imposes $n$ different local acyclicity conditions,
which is weaker than the condition that $L_+(\mathcal{K}_S)$ be (globally) acyclic. Global acyclicity is equivalent to concordance (in the ranking system language). For an example, see \cite{barex}.

\begin{proof}
The equivalence was explained in (\ref{e:direction}) above. 
Discussion above shows that a transitive tournament on $V_x$ is equivalent to a total order $\prec_x$ on $S$.
The Autosimilarity requirements (4) of Section \ref{s:triplets}, and (IV) of Section \ref{s:totalorder},
are equivalent to the requirement that $\{x, x\}$ be a source.
\end{proof}


\section{Local depth and the cluster graph}\label{s:pald}

\subsection{Non-parametric approach to low-dimensional embedding and partitioning}
Section \ref{s:pald} will be an overview of the non-parametric
\textbf{partitioned local depth} (\texttt{PaLD}) algorithm by
Berenhaut, Moore and Melvin \cite{ber}. \texttt{PaLD} converts
a family of partial orders $\{\preceq_x, x \in S\}$ into 
a matrix $(C_{x,y})_{x,y \in S}$ of ``cohesion scores'',
where for fixed $x \in S$, the row vector $(C_{x,y})_{y \in S}$
represents a disaggregation of a quantity called the ``local depth''
of $x$, defined below. The authors place the weight 
\[
w_{x, y}:=\min \{ C_{x,y}, C_{y,x} \}
\]
on each edge of the complete graph on $S$, and use standard
force-directed graph embedding algorithms \cite{fru} to
present a \textit{planar embedding} of $S$ for the purpose of visual comprehension.
\texttt{PaLD} has ontological priority over methods
which start with an explicit vector representation of the points in $S$,
such as t-SNE \cite{maa} and UMAP \cite{mci}.
Suppose there are two different metrics $\rho$ and $\rho'$ (or vector representations), which induce the
same ranking system, i.e. $y \preceq_x z$ whenever $\rho(x, y) \leq \rho(x, z)$, etc.
\begin{enumerate}
    \item[(i)] Planar embedding supplied by vector-based methods may depend on the metric.
    \item[(ii)] \texttt{PaLD} gives an embedding depending only on the ranking system, not the metric.
\end{enumerate}

The cohesion scores supply a natural \textbf{cluster threshold}
\begin{equation} \label{e:clusterthreshold}
    \tau:=\frac{1}{2n|} \sum_{x \in S} C_{x,x},
\end{equation}
which Corollary \ref{c:meanlocaldepth} will show to be the mean local depth, divided by $n$.
A simple, undirected, unweighted \textbf{cluster graph} $G:=(S, E)$
has edge set $E$ consisting of those pairs $\{x, y\}$ for which
$w_{x, y} \geq \tau$. The components of this graph\footnote{
An alternative partitioning arises from the strongly-connected components 
of the directed graph on $S$, consisting of those
arcs $x\rightarrow y$ for which $C_{x,y} \geq \tau$. This appears less useful in practice.
} supply a
non-parametric partitioning of $S$. 

In the case where no ties are allowed, here are explicit definitions, following \cite{ber},
but in terms of a family of total orders as in Section \ref{s:totalorder}.


\begin{table}  
\centering
\begin{tabular}{ |c|c|c|c|p{.22\linewidth}| } 
 \hline
  Ranking digraph & Relative position & $z \in  U_{x, y}$? & $z \in  U_{x \| y}$? & Comments \\
   
   \hline
   
 \( \begin{tikzcd}{} & xz \arrow{dr}{} \\xy \arrow{ur}{} \arrow{rr}{} && yz\end{tikzcd} \)
 & $\stackrel{y}{\bullet} \quad \stackrel{x}{\bullet} \qquad \stackrel{z}{\bullet}$
 & \textsf{x} & \textsf{x} & 
 \multirow{4}{\linewidth}{Here, $z$ is ``off in the wilderness'' relative to $x$ and $y$ (i.e.\ $z \notin U_{x,y}$).} 
 \\  
 
 \( \begin{tikzcd}{} & yz \arrow{dr}{} \\xy \arrow{ur}{} \arrow{rr}{} && xz\end{tikzcd} \)&
 $\stackrel{x}{\bullet} \quad \stackrel{y}{\bullet} \qquad \stackrel{z}{\bullet}$ & 
 \textsf{x} & \textsf{x} & \\  
 
 \hline

  \( \begin{tikzcd}{} & xy \arrow{dr}{} \\yz \arrow{ur}{} \arrow{rr}{} && xz\end{tikzcd} \)&
  $\stackrel{z}{\bullet} \quad \stackrel{y}{\bullet} \qquad \stackrel{x}{\bullet}$ & 
   \checkmark & \textsf{x} & 
  \multirow{4}{\linewidth}{Here, $z$ is in the ``sphere of influence'' of $x$ and $y$ (i.e.\ $z \in U_{x,y}$), and from $z$'s point of view, $y$ is more similar than $x$ (so $z \notin U_{x\|y}$).} 
  \\ 
 
  \( \begin{tikzcd}{} & xz \arrow{dr}{} \\yz \arrow{ur}{} \arrow{rr}{} && xy\end{tikzcd} \)&  $\stackrel{y}{\bullet} \quad \stackrel{z}{\bullet} \qquad \stackrel{x}{\bullet}$ & 
 \checkmark & \textsf{x} & \\ 

  \( \begin{tikzcd}{} & xz \arrow{dr}{} \\yz \arrow{ur}{} \arrow[leftarrow]{rr}{} && xy\end{tikzcd} \) &
  \begin{tabular}{c}
    no metric embedding \\ (has a cycle)
  \end{tabular} &
  \checkmark & \textsf{x} & \\ 

 \hline
 
 \( \begin{tikzcd}{} & xy \arrow{dr}{} \\xz \arrow{ur}{} \arrow{rr}{} && yz\end{tikzcd} \)&
  $\stackrel{z}{\bullet}  \quad \stackrel{x}{\bullet} \qquad \stackrel{y}{\bullet}$ & 
   \checkmark & \checkmark & 
  \multirow{4}{\linewidth}{Here, $z$ is in the ``sphere of influence'' of $x$ and $y$ (i.e.\ $z \in U_{x,y}$), and from $z$'s point of view, $x$ is more similar than $y$ (so $z \in U_{x\|y}$).} 
  \\ 
 
  \( \begin{tikzcd}{} & yz \arrow{dr}{} \\xz \arrow{ur}{} \arrow{rr}{} && xy\end{tikzcd} \)&
  $\stackrel{x}{\bullet} \quad \stackrel{z}{\bullet} \qquad \stackrel{y}{\bullet}$ & 
  \checkmark & \checkmark & \\ 
  
  \( \begin{tikzcd}{} & yz \arrow{dr}{} \\xz \arrow{ur}{} \arrow[leftarrow]{rr}{} && xy\end{tikzcd} \) &
  \begin{tabular}{c}
    no metric embedding \\ (has a cycle)
  \end{tabular} &
  \checkmark & \checkmark & \\ 
  
   \hline
\end{tabular}

\caption{\textit{Conflict focus membership: } 
The first column shows the eight orientations of the line graph of the complete graph on $\{x,y,z\}$.
When possible, the second column shows an instantiation of this ranking digraph with $x,y,z$ as points on a Euclidean line. The third and fourth indicate $z$'s membership in $U_{x,y}$ and $U_{x\|y}$ in each case.}
\label{t:conflictfocus}
\end{table}


\subsection{Local depth} \label{s:conflictfocus}
For distinct $x, y \in S$, let
\begin{equation}\label{e:leftconflictfocus}
  U_{x \| y}:=  \{z \in S: z \prec_x y \mbox{ and } x \prec_z y \} \supset \{x\}.  
\end{equation}
The inclusion of $x$ in (\ref{e:leftconflictfocus}) follows from Autosimilarity.
In the augmented line graph, $z \in U_{x \| y}$ corresponds to arcs
\(
\{y, z\} \longleftarrow \{x, z\} \longrightarrow \{x, y\}.
\)
The symmetrized version is
\begin{equation}\label{e:conflictfocus}
    U_{x, y}:=U_{x \| y} \cup U_{y \|   x} = \{z \in S: z \prec_x y  \} \cup \{t \in S: t \prec_y x \} 
    \supset \{x, y\}.
\end{equation}
and is called the \textbf{conflict focus}\footnote{By constrast, 
\cite{kle} defines the
\textbf{lens} of $\{x, y\}$ to consist of
the set of $z$ for which $z \preceq_x y$ and $z \preceq_y x$.}.
 The first union in (\ref{e:conflictfocus}) is disjoint,
 while the second need not be (Table \ref{t:conflictfocus}).
Comparing (\ref{e:direction}) and (\ref{e:conflictfocus}) shows that
the size $|U_{x,y}|$ of the conflict focus is the in-degree of the vertex $\{x, y\}$
in the augmented line graph for this ranking system.
 These concepts are illustrated in Table \ref{t:conflictfocus}. 
 The reason for splitting the conflict focus in two is apparent in  (\ref{e:ldep}).
 
\begin{defn} \label{d:ldep}
Fix $x \in S$, sample $Y$ uniformly from $S \setminus \{x\}$, and then
sample $Z$ uniformly from the conflict focus $U_{x, Y}$. Define the
\textbf{local depth} at $x$ to be\footnote{The random tie breaker in \cite{ber} has been omitted,
since we assume there are no ties.}
\begin{equation} \label{e:ldep}
\ell(x):= \p[x \prec_Z Y] = \E\left(\frac{|U_{x \| Y}| }{ |U_{x, Y}| }\right).
\end{equation}
\end{defn}
Interpret (\ref{e:ldep}) as asking
a uniform random $Z$ in the conflict focus $U_{x, Y} = U_{Y, x}$:
\textit{which of $Y$ and $x$ is more similar to you?}
The local depth of $x$ is the probability the ``witness'' $Z$
chooses $x$ rather than the random alternative $Y$. 

\subsection{Cohesion}
Let us partition local depth according to the
value taken by the random variable $Z$ in (\ref{e:ldep}).

\begin{defn} \label{d:cohesion}
For any elements $x, v \in S$ (possibly identical), the \textbf{cohesion}
of $v$ to $x$ is defined as
\begin{equation} \label{e:cohesion}
C_{x,v}:= \p[\{x \prec_Z Y\} \cap \{Z = v\}].
\end{equation}
Call $(C_{x,v})_{x, v \in S}$ the \textbf{cohesion matrix}.
\end{defn}

The following formulas are implicit in the methods of \cite{ber},
simplified by the absence of ties. 

\begin{prop} \label{p:explcohe}
The cohesion may be written in the explicit form:
\begin{equation} \label{e:cohesionx} 
C_{x,v} = \frac{1}{n-1}
\sum_{y \neq x} \frac{ 1_{ \{v \in  U_{x \| y} \} } }{ |U_{x ,y}|} .
\end{equation}
In particular,
\[
C_{x,x} = \frac{1}{n-1} \sum_{y \neq x} \frac{1}{|U_{x, y}|}
= \E\left[ \frac{1}{|U_{x,Y}|} \right],
\]
where $Y$ is sampled uniformly at random from $S \setminus \{x\}$. 
Sample an unordered pair $\{X, Y\}$ uniformly at random from $\binom{S}{2}$.
The cluster threshold (\ref{e:clusterthreshold}), i.e. $\tau:=(\sum_x C_{x,x})/(2n)$,
may be written
\begin{equation} \label{e:meancohesion}
\tau = \E\left[ \frac{1}{|U_{X,Y}|} \right].
\end{equation}
\end{prop}

\begin{proof}
Condition on the event $\{Y = y\}$, which has 
probability $\frac{1}{n-1}$ for $y \neq x$. 
In that case, the random variable $Z$ equals $v \in U_{x, y}$
with probability $1/|U_{x, y}|$. Hence
\[
C_{x,v} = \sum_{y \neq x}
\frac{\p[Y = y]}{|U_{x, y}|} 1_{ \{ x \prec_v y \} },
\] 
This may be rewritten in the form (\ref{e:cohesionx}).
The assertion about $C_{x,x}$ follows from Autosimilarity:
$x \prec_x y$ for all $y \neq x$.
If we sample $X$ uniformly from $S$, then sample $Y$ uniformly from $S \setminus \{X\}$,
\(
\tau = \frac{1}{2} \E[ C_{X,X} ]
\)
which gives (\ref{e:meancohesion}), because half the sum over ordered pairs $(x, y)$
is the sum over unordered  pairs $\{x, y\}$.
\end{proof}

\begin{cor}\label{c:meanlocaldepth}
The cluster threshold $\tau$ is $1/n$ times the mean value of the local depth.
\end{cor}

\begin{proof}
Sample $X$ uniformly at random from $S$. By definition of $C_{x,v}$,
\begin{equation}\label{e:meanlocaldepth}
  \E[\ell(X)] = 
\frac{1}{n} \sum_x \sum_{v \neq x} C_{x,v} + \frac{1}{n} \sum_x C_{x,x}.  
\end{equation}
By (\ref{e:cohesionx}), we may write the first sum as
\[
\sum_x \sum_{v \neq x} C_{x,v} = 
 \sum_{x} \sum_{v \neq x}  \frac{1}{n-1} \sum_{y \notin \{x, v\}}
\frac{1_{ \{ x \prec_v y \} } }{ |U_{x, y}|}  = 
\frac{1}{n-1}\sum_{x} \sum_{y \neq x} \frac{1}{ |U_{x, y}|} \sum_{v \notin \{x, y\}} 1_{ \{ x \prec_v y \} }.
\]
We exclude $y = v$ from the middle sum because the relation $x \prec_y y$ is false.
The formula is unchanged if the symbols $x$ and $y$ are exchanged, and so the final sum may be replaced by $(n-2)/2$.
The expression becomes:
\[
\frac{n-2}{2(n-1)} \sum_{\substack{x, y\\x \neq y }} \frac{1}{ |U_{x, y}|}
= \frac{n-2}{2} \sum_x C_{x,x}.
\]
Thus (\ref{e:meanlocaldepth}) becomes:
\[
\E[\ell(X)] = 
\frac{n-2}{2 n} \sum_x C_{x,x} + \frac{1}{n} \sum_x C_{x,x} = \frac{1}{2} \sum_x C_{x,x}.
\]
The result follows from (\ref{e:clusterthreshold}).
\end{proof}

\subsection{Complexity of partitioned local depth} 
Complexity of \texttt{PaLD} is
a barrier to scalable implementation, and the main motivator for later parts of this paper.
A total order $\preceq_x$ allows each $y \in S \setminus \{x\}$ to be assigned an integer rank, namely
$| \{z \in S: z \prec_x y \} |$.
For each choice of $x$, sorting the $n-1$ elements of $S \setminus \{x\}$ needs
$n \log{n}$ evaluations of the \texttt{Comparator}
at $x$ (using the \texttt{sort} method of Java \texttt{List} \cite{javalist}, for example), 
so computing all the ranks $(r_x(y))_{x, y \in S}$
is an $O(n^2 \log{n})$ task. We consult the rank tables when we need
to know later whether $u \prec_v w$: check whether $r_v(u) < r_v(w)$. This one time work avoids making
$O(n^3)$ calls to the \texttt{Comparator}.

\begin{prop}\label{p:cubic}
Computing the sizes $\{|U_{x, y}|, \,\{x, y\} \in \binom{S}{2} \}$ of all the $\binom{n}{2}$ conflict foci,
and the $n \times n$ cohesion matrix $(C_{x,v})_{x, v \in S}$, can be performed in
 $O(n^3)$ operations.
\end{prop}
\textbf{Remark: }Although computation of the cohesion matrix may be 
distributed among multiple processors, we do not know of a sub-cubic
method to compute it. The cohesion matrix is needed to determine the cluster graph, 
whose edge set consists of those pairs $\{x, y\}$ for which 
$\min \{ C_{x,y}, C_{y,x} \}\geq \tau$.
Partitioning via local depth appears to require a run time which is
cubic in the size of the set to be partitioned.

\begin{proof}
For each pair $\{x, y\}$, we must determine membership of $U_{x \| y}$ and $U_{y \| x}$, then
count the elements of the disjoint union
\[
U_{x, y}:=U_{x \| y} \cup U_{y \| x},
\]
For this a loop through $S\setminus \{x, y\}$ suffices, which is $O(n)$ work for each of $O(n^2)$
pairs $\{x, y\}$, giving $O(n^3)$ operations.
Proposition \ref{p:explcohe} shows that a similar $O(n^3)$
procedure suffices to compute the sum on the right side of
(\ref{e:cohesionx}), inserting the precomputed coefficients
$(|U_{x, y}|)_{x, y \in S}$.
\end{proof}

Our goal is to approximate \texttt{PaLD} at less than 
the $O(n^3)$ cost described in Proposition \ref{p:cubic}.


\section{Local depth restricted to neighbors}\label{s:pannld}

\begin{figure}
\caption{\textbf{Nearest neighbors digraph:} 
\textit{neighbors of a single vertex $v$ in a 3-nearest-neighbors digraph,
and arcs between them, are highlighted.
In this instance, the 3-nearest-neighbors digraph split into components of size 15 and 5,
respectively.
}
}
\label{f:nndigraph}
\begin{center}
\scalebox{0.5}{\includegraphics{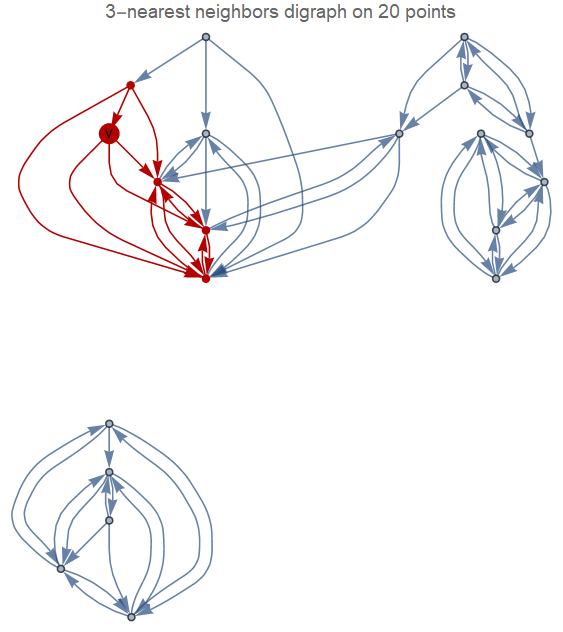} }
\end{center}

\end{figure}

\subsection{Several plausible methods}
When attempting to adapt the methods of Section \ref{s:pald} to the $K$-nearest neighbor digraph,
several approaches come to mind, such as:
\begin{enumerate}
    \item [(A)] Modify Definition \ref{d:ldep} of local depth so $Y$ is sampled only from the $K$-nearest neighbors
    of $x$, and restrict $Z$ to $K$-nearest neighbors of $x$ or $Y$.
    
    \item [(B)]
    Treat a triplet comparison $y \prec_x z$ as a tie when neither $y$ nor $z$ is among the 
     $K$-nearest neighbors of $x$.
     
    \item [(C)]
    Randomization of triplet comparisons which are not made explicitly.
\end{enumerate}
We investigated all three\footnote{
Method (A) lacks a clear rationale for selecting the threshold (\ref{e:clusterthreshold}).
Method (B) does not lend itself to the kind of sensitivity analysis offered in Theorem \ref{t:cohesionconcentration}, and
the resulting clustering was less effective in practice.
}, but only (C) gave a satisfactory theory. This is presented below.

\subsection{Friends, strangers, and promoted pairs}\label{s:friends}
Suppose $\{\preceq_x, x \in S \}$ is a ranking system, and $K$ is an integer with
$1 < K < n-1$. For each $x \in S$, specify an integer $K_x$ such that $K \leq K_x < n-1$,
and a subset $\Gamma_x \subset S \setminus \{x\}$ of size $K_x$ with the property that
\begin{equation}\label{e:friendsprecede}
    y \prec_x z, \quad \forall y \in \Gamma_x, \quad \forall z \in S \setminus (\{x\} \cup \Gamma_x).
\end{equation}
Refer to $\Gamma_x$ as $K_x$-nearest neighbors, 
or simply as \textbf{friends} of $x$, and to 
\(
\Delta_x:=S \setminus (\{x\} \cup \Gamma_x)
\)
as the \textbf{strangers} of $x$; (\ref{e:friendsprecede}) says 
that the total order $\preceq_x$ places any friend of $x$ before any stranger of $x$.
When $x \in \Gamma_y$, we say that $y$ is a co-friend of $x$. 

In some applications to unsupervised learning, $K_x = K$, for all $x$, and
$\Gamma_x$ consists of the $K$ nearest neighbors of $x$ according to $\preceq_x$.
In other applications (e.g. Section \ref{s:bipartite}), $S$ arises as the vertex set of a graph,
$\Gamma_x$ consists of the vertices adjacent to $x$, and these neighbors precede
the non-neighbors of $x$ according to $\preceq_x$. 

Given friend sets $\{ \Gamma_x, x \in S \}$ as in (\ref{e:friendsprecede}),
the \textbf{promoted pairs} consist of the subset
\begin{equation}\label{e:promotedpairs}
    \mathcal{P}:= \{ \{x, y\}: y \in \Gamma_x \, \mbox{ or } \, x \in \Gamma_y \} 
    \subset \binom{S}{2}.
\end{equation}
In other words, if $y$ is a friend or co-friend of $x$, then $\{x, y\}$ is a
promoted pair\footnote{Contrast $\mathcal{P}$ with the set of mutual friends:
\(
    \mathcal{B}:= \{ \{x, y\}: y \in \Gamma_x \, \mbox{ and } \, y \in \Gamma_x \}.
\)
In the nearest neighbors digraph (Figure \ref{f:nndigraph})
$\mathcal{B}$ is the set of vertex pairs connected by edges in both directions,
whereas $\mathcal{P}$ is the edge set of the undirected graph. In terms of $\bar{K}:= (\sum_{x \in S} K_x)/n,$
\(
0 \leq |\mathcal{B}| \leq n \bar{K}/2 \leq |\mathcal{P}| \leq n \bar{K},
\)
and each inequality is sharp, meaning that there is an example where equality holds.
An approximation to \texttt{PaLD} could be formulated either in terms of $\mathcal{B}$ of $\mathcal{P}$,
but $\mathcal{B}$ may be quite sparse, whereas $\mathcal{P}$ grows with $K$ in a way that enables
concentration inequalities.
}.

\subsection{Restricting the local depth computation to promoted pairs}
The complement of the set of promoted pairs is called the set of \textbf{relegated pairs}:
\begin{equation}\label{e:relegation}
    \mathcal{R}:= \binom{S}{2} \setminus \mathcal{P} = 
    \{ \{x, y\}: y \in \Delta_x \, \mbox{ and } \, y \in \Delta_x \}.
\end{equation}
Introduce the notations:
\[
\mathcal{P}_x:= \{y \in S \setminus \{x\}: \{x, y\} \in \mathcal{P} \}; \quad
\mathcal{R}_x:= \{y \in S \setminus \{x\}: \{x, y\} \in \mathcal{R} \}.
\]
Refer to Figure \ref{f:nndigraph} for an example in which vertex $v$ 
and elements of $\mathcal{P}_v$ are highlighted, as well as directed edges between members of $\mathcal{P}_v$.
The sizes of these sets satisfy: $d_x:=|\mathcal{P}_x| \geq K_x$, and $|\mathcal{R}_x| \leq n - K_x - 1$.
There is no universal upper bound on $d_x$ in terms of $K$, because
one specific point can be the nearest neighbor of all the other points: see
Section \ref{s:stargraph} for a natural example.
Let $\bar{K}$ denote the average $(\nicefrac{1}{n}) \sum_x K_x$; useful inequalities include:
\begin{equation}\label{e:promographdegrees}
\frac{1}{2} n \bar{K} \leq
    |\mathcal{P}| = \frac{1}{2}\sum_{x \in S} d_x \leq \sum_{x \in S} K_x =  n \bar{K}.
\end{equation}
The first holds because  $K_x \leq d_x$, and the second because the total number of directed edges is at least
half the total undirected degree.

Our approximation to local depth will rest on three ingredients:
\begin{enumerate}
    \item[(a)] 
    Perform explicit triplet comparisons of the form ``$y \prec_v z$'' 
    only when both $\{v, y\}$ and $\{v, z\}$ are promoted pairs. 
    See highlighted edges in Figure \ref{f:nndigraph}, treating $v$
    as the central highlighted vertex.
    \item[(b)] 
    Introduce a deliberate distortion into the ranking system by
treating every $y \in \mathcal{P}_x$ as more similar to $x$ than any
$z \notin \mathcal{P}_x$, even when the ranking $\preceq_x$
would decide otherwise. For any fixed $x, y, z$, the distortion
disappears when $K_x$ is sufficiently large.

    \item[(c)] 
    When both $\{x, z\}$ and $\{y, z\}$ are relegated pairs, a randomization procedure is implicitly invoked, and
    corresponding terms in the cohesion formula will be replaced by averages with respect to the randomization,
    as we shall describe in Section \ref{s:randomrelegates}.

\end{enumerate}

\subsection{Conflict focus in terms of promoted pairs}

The restrictions (a) and (b) lead to a a subtle reformulation of conflict focus.
First we introduce the counterpart to (\ref{e:leftconflictfocus}) .

\begin{defn}\label{d:leftpromoconflictfocus}
For $\{x, y\} \in \mathcal{P}$, the \textbf{half-focus} is
\[
      U^{\mathcal{P}}_{x \| y}:=  \{x\} \cup \{z \in \mathcal{P}_x \cap \mathcal{P}_y: z \prec_x y \mbox{ and } x \prec_z y \}
      \cup \{z \in \mathcal{P}_x \setminus \mathcal{P}_y: z \prec_x y \}.
\]
\end{defn}

To reiterate (a), the comparison ``$x \prec_z y$'' is allowed only when
both $\{x, z\}$ and $\{y, z\}$ are promoted pairs,
i.e. $z \in \mathcal{P}_x \cap \mathcal{P}_y$. When $\{x, z\}$ is
promoted but $\{y, z\}$ is relegated,
i.e. $z \in \mathcal{P}_x \setminus \mathcal{P}_y$,
then $z$ prefers $x \in \mathcal{P}_z$.
The distortion (b) is deliberately maintained: 
$x \in \mathcal{P}_z$ is considered more similar to $z$ than
any $y \notin \mathcal{P}_z$, whatever the ranking says. We are not defining $U^{\mathcal{P}}_{x \| y}$ at all
when $\{x, y\}$ is a relegated pair.

\begin{defn}\label{d:promoconflictfocus}
For each $\{x, y\} \in \mathcal{P}$, the \textbf{promoted conflict focus} $U^{\mathcal{P}}_{x, y}$ is defined by:
\[
    U^{\mathcal{P}}_{x, y}:= U^{\mathcal{P}}_{x \| y} \cup U^{\mathcal{P}}_{y \| x} = 
    \{x, y\} \cup \{z \in \mathcal{P}_x : z \prec_x y  \} \cup \{z' \in \mathcal{P}_y: z' \prec_y x \}.
\]
\end{defn}
As with (\ref{e:conflictfocus}), $U^{\mathcal{P}}_{x \| y}$ and $U^{\mathcal{P}}_{y \| x}$ are disjoint, with
$x \in U^{\mathcal{P}}_{x \| y}$ and $y \in U^{\mathcal{P}}_{y \| x}$.

\subsection{Stranger randomization}\label{s:strangerrand}
The following construction is crucial to our approximation to
local depth. The idea is not to invoke triplet comparisons ``$y \prec_x z$''
when neither $\{x, y\}$ and $\{x, z\}$ is a promoted pair,
but to give an average answer over random concordant ranking systems.
Randomization is a theoretical device, used only for deriving formulas, not
for generating random samples. Let 
\begin{equation}\label{e:randomizerelegated}
    \eta:=\{\eta_{x,y}, \{x, y\} \in \mathcal{R} \}.
\end{equation}
be a collection of i.i.d. Uniform$(0, 1)$ random variables. Restrict to the probability one event that no ties occur.
Our approximation will replace the ranking system $\{\preceq_x, x \in S\}$ by its \textbf{stranger randomization}
$\{\preceq^{\eta}_x, x \in S\}$, where the relation $y \prec^{\eta}_x z$ means one of three things:
\begin{enumerate}
    \item[(A)] 
    $y, z \in \mathcal{P}_x$ and $y \prec_x z$: ``usual comparison
    of two promoted pairs''.
   
    \item[(B)] 
    $\{x, y\} \in \mathcal{P}$ and $\{x, z\} \in \mathcal{R}$: ``promoted beats relegated''.
    
    \item[(C)] 
    $z \in \mathcal{R}_x \cap \mathcal{R}_y$ and $\eta_{x,y} <  \eta_{x,z}$: ``stranger randomization''.
    
\end{enumerate}

As a complement to Definition \ref{d:leftpromoconflictfocus} of half-focus for promoted pairs, here is the
random counterpart to (\ref{e:leftconflictfocus}) for relegated pairs.

\begin{defn}\label{d:leftrelegconflictfocus}
For $\{x, y\} \in \mathcal{R}$, the half-focus is the random set:
\[
      U^{\mathcal{R}}_{x \| y}(\eta):=  \{x\} \cup (\mathcal{P}_x \setminus \mathcal{P}_y) \cup
      \{z \in \mathcal{P}_x \cap \mathcal{P}_y: x \prec_z y \}
      \cup \{z \in \mathcal{R}_x \cap \mathcal{R}_y: \eta_{x,z} < \min \{  \eta_{x,y}, \eta_{y,z} \} \}.
\]
\end{defn}

It is possible for $\mathcal{P}_x \cap \mathcal{P}_y$ to be non-empty,
even when $\{x, y\}$ is relegated. When both $\{x, z\}$ and $\{y, z\}$ 
are relegated, the random condition $\eta_{x,z} <  \eta_{x,y}$ replaces $z \prec_x y$, and the random condition
$\eta_{x,z} <  \eta_{y, z}$ replaces $x \prec_z y$. 
Following the pattern of Definition \ref{d:promoconflictfocus}, for promoted pairs,
here is the notion of conflict focus for relegated pairs.

\begin{defn}\label{d:relegconflictfocus}
For $\{x, y\} \in \mathcal{R}$ the \textbf{relegated conflict focus} $U^{\mathcal{R}}_{x, y}(\eta)$ is defined by:
\[
    U^{\mathcal{R}}_{x, y}(\eta):= U^{\mathcal{R}}_{x \| y}(\eta) \cup U^{\mathcal{R}}_{y \| x}(\eta) = 
    \{x, y\} \cup \mathcal{P}_x \cup \mathcal{P}_y \cup
    \{z \in \mathcal{R}_x \cap \mathcal{R}_y: \eta_{x,y} > \min \{  \eta_{x,z}, \eta_{y,z} \} \}.
\]
\end{defn}

\subsection{Cohesion matrix in terms of promoted and relegated pairs}

In the context of promoted and relegated pairs, and the stranger randomization
$\{\preceq^{\eta}_x, x \in S\}$, the local depth computation goes as follows. 
Fix $x \in S$, and sample $Y$ uniformly from $S \setminus \{x\}$. There are two possibilities:
\begin{enumerate}
    \item[(1)] 
    If $\{x, Y\}$ is a promoted pair, sample $Z$ uniformly from 
the promoted conflict focus $U^{\mathcal{P}}_{x, Y}$. Then the local depth at $x$
is the probability $\p[Z \in U^{\mathcal{P}}_{x \| Y} ]$.

    \item[(2)] 
     If $\{x, Y\}$ is a relegated pair, sample $Z$ uniformly from the relegated conflict focus
 $U^{\mathcal{R}}_{x, Y}(\eta)$. Then the (random) local depth at $x$
is the conditional probability $\p[Z \in U^{\mathcal{R}}_{x \| y}(\eta) ]$,
given $\eta$.
    
\end{enumerate}

According to this sampling procedure,
the (random) cohesion matrix (\ref{e:cohesionx})
for the randomized ranking system $\{\preceq^{\eta}_x, x \in S\}$,
denoted $(C_{x,v}^{\eta})$, is the sum of contributions from promoted and relegated pairs:
\begin{equation}\label{e:randomcohesion}
C_{x,v}^{\eta}:=C_{x,v}^{\mathcal{P}} + C_{x,v}^{\mathcal{R}}(\eta); \quad
v \in \{x\} \cup \mathcal{P}_x ,
\end{equation}
where in terms of Definitions \ref{d:leftpromoconflictfocus}, \ref{d:promoconflictfocus},
\ref{d:leftrelegconflictfocus}, and \ref{d:relegconflictfocus},
\[
C_{x,v}^{\mathcal{P}} := \frac{1}{n-1}
\sum_{y \in \mathcal{P}_x} 
\frac{ 1 } { |U^{\mathcal{P}}_{x, y}| } 1_{ \{v \in U^{\mathcal{P}}_{x \| y} \} };
\quad
C_{x,v}^{\mathcal{R}}(\eta):=\frac{1}{n-1}
\sum_{y \in \mathcal{R}_x} 
\frac{ 1 } { |U^{\mathcal{R}}_{x, y}(\eta)| }  1_{ \{v \in U^{\mathcal{R}}_{x \| y}(\eta) \} }.
\]
\subsubsection{Comments on the random cohesion matrix}
\begin{enumerate}
    \item[a.] 
    Apart from the novelty of stranger randomization, (\ref{e:randomcohesion}) is wholly consistent with
    the methods of Section \ref{s:pald}.
        \item[b.]
        The deterministic first term $C_{x,v}^{\mathcal{P}}$, called the \textbf{promoted cohesion matrix},
        will be computed in full detail.
            \item[c.]
            The relegated cohesion matrix $C_{x,v}^{\mathcal{R}}(\eta)$ is a random
function of (\ref{e:randomizerelegated}), and will not be computed
explicitly, but merely estimated via its mean.
\end{enumerate}

\subsubsection{Expectation of the random cohesion matrix}
Definition \ref{d:pannldmatrix} is phrased in terms of the nearest neighbors
digraph $(S, F)$, where $F$ is the union of the
arcs from $x$ to each of its $K_x$ nearest neighbors, over $x \in S$.

\begin{defn}\label{d:pannldmatrix}
A pair $\{x, y\} \in \mathcal{P}$ (i.e is promoted) if $x$ is one of the $K_y$ nearest
neighbors of $y$, or $y$ is one of the $K_x$ nearest 
neighbors of $x$. The PaNNLD cohesion matrix $(C_{x,v}^{F})$, defined only when $v = x$ or
$\{x, v\} \in \mathcal{P}$, is derived from (\ref{e:randomcohesion}):
\begin{equation}\label{e:pannldmatrix}
    C_{x,v}^{F}:= C_{x,v}^{\mathcal{P}} + \E[C_{x,v}^{\mathcal{R}}(\eta)]
\end{equation}
where expectation is taken with respect to the family $\eta$ as in (\ref{e:randomizerelegated}).
\end{defn}

Three main results hold for (\ref{e:pannldmatrix}), whose proofs will
occupy the rest of this paper. Efficiently computable expressions for the terms on
the right side of (\ref{e:pannldmatrix}) appear in Theorems \ref{t:complexity} and
\ref{t:averages}, respectively. The quality of approximation of the PaNNLD cohesion
to the \texttt{PaLD} cohesion is estimated in Theorem \ref{t:cohesionconcentration}.

\subsection{Complexity result}
To evaluate the promoted cohesion matrix appearing in (\ref{e:randomcohesion}) and
(\ref{e:pannldmatrix}), we use the fact, implied by Definition \ref{d:promoconflictfocus}, that
\[
|U^{\mathcal{P}}_{x, y}| = |U^{\mathcal{P}}_{x\| y}| + |U^{\mathcal{P}}_{y\| x}|.
\]
It will suffice to establish membership of all half-foci 
$\{U^{\mathcal{P}}_{x\| y}: \{x, y\} \in \mathcal{P} \}$.

\begin{thm}[COMPLEXITY]\label{t:complexity}
Given a ranking system on $S$, and a nearest neighbors
digraph $(S, F)$ with undirected version $(S, \mathcal{P})$,
there is a combinatorial algorithm which computes $C_{x,v}^{\mathcal{P}}$, for all $\{x, v\} \in \mathcal{P}$
and for $x = v$.
\begin{enumerate}
    \item [(a)]
    The algorithm requires about
    \(
\sum_{x \in S}  d_x \log{d_x} \,( \leq 2 n \bar{K} \log{n} )
\)
triplet comparisons\footnote{Sum over $x$ the cost of sorting $\mathcal{P}_x$ with respect to $\prec_x$.
See the \texttt{sort} method of Java \texttt{List} \cite{javalist}},
where $d_x$ is the degree of $x$ in $(S, \mathcal{P})$.  
    \item [(b)]
    The number of steps in the algorithm's inner loops does not exceed 
        \(
(\nicefrac{3}{2}) \sum_{x \in S}  d_x^2 
\).
\end{enumerate}
\end{thm}
The proof will be given in Section \ref{s:promocomplexity}, using
Algorithm \ref{a:promocohesion}.

\subsection{Randomization averages}
On comparing Definition \ref{d:leftrelegconflictfocus} and (\ref{e:randomcohesion}), we see that
\begin{equation}\label{e:randomdiagonal}
   (n-1) C_{x,x}^{\mathcal{R}}(\eta) = \sum_{y \in \mathcal{R}_x} \frac{1}{|U^{\mathcal{R}}_{x, y}(\eta)|}.
\end{equation}
The mean of the expression (\ref{e:randomdiagonal}) will be easy to estimate.
Then treat $C_{x,v}^{\mathcal{R}}(\eta)$, for $\{x, v\} \in \mathcal{P}$,
as a perturbation of (\ref{e:randomdiagonal}):
adjust the last expression when
$y \in \mathcal{P}_v$, i.e. $v \in  \mathcal{P}_x \cap \mathcal{P}_y$,
to take account of whether or not $x \prec_v y$:
\begin{equation}\label{e:randompromoted}
     C_{x,v}^{\mathcal{R}}(\eta) = C_{x,x}^{\mathcal{R}}(\eta) -  \frac{1}{n-1}
     \sum_{\substack{y \in \mathcal{P}_v \cap \mathcal{R}_x\\
     y \prec_v x} }
     \frac{1}{|U^{\mathcal{R}}_{x, y}(\eta)|}, \quad \{x, v\} \in \mathcal{P}.
\end{equation}

For $\{x, y\} \in \mathcal{R}$, define the \textbf{range of influence}:
\begin{equation}\label{e:rangeofinfluence}
    m_{x, y}:=|\{x, y\} \cup \mathcal{P}_x \cup \mathcal{P}_y| =
   2 + d_x + d_y - | \mathcal{P}_x \cap \mathcal{P}_y | = n - |\mathcal{R}_x \cap \mathcal{R}_y|.
\end{equation}

The next major result gives an explicit formula for the
the second term of (\ref{e:pannldmatrix}), which avoids 
computing expected values separately for the $O(n^2)$
relegated pairs. 

\begin{thm}[AVERAGE COHESION]\label{t:averages}
A function $\phi_n:\{2, 3, \ldots, n\} \to \R_+$ exists such that
\[
\E\left( 1/|U^{\mathcal{R}}_{x, y}(\eta)| \right) = \phi_n(m_{x, y}),
\]
where expectation is taken with respect to $\eta$ as in (\ref{e:randomizerelegated}).
The expected value of the relegated cohesion matrix in (\ref{e:pannldmatrix})
may be expressed using (\ref{e:randomdiagonal}) and (\ref{e:randompromoted}):
\begin{align*}
 \E[C_{x,x}^{\mathcal{R}}(\eta)] =&  \sum_{r = 2}^n \phi_n(r) \cdot
\frac{|\{y \in \mathcal{R}_x: \, m_{x, y} = r\}|  }{n-1}  \\
 \E[C_{x,v}^{\mathcal{R}}(\eta)] =& \E[C_{x,x}^{\mathcal{R}}(\eta)] - 
  \frac{1}{n-1}
     \sum_{\substack{y \in \mathcal{P}_v \cap \mathcal{R}_x\\
     y \prec_v x} }
     \phi_n(m_{x, y}), \quad \{x, v\} \in \mathcal{P}.
\end{align*}
All these sums are computable in time proportional to the sum of squares of the 
vertex degrees of the graph $(S, \mathcal{P})$. 
\end{thm}

\textbf{Remarks:}
\begin{enumerate}
    \item The proof will be given in Section \ref{s:averages}, which supplies efficient algorithms,
    depending only on the graph $(S, \mathcal{P})$, for counting
    $\{y \in \mathcal{R}_x: \, m_{x, y} = r\}$ when $\{x, y\} \in \mathcal{R}$.
    
    \item  $\phi_n(n) = 1/n$, and for large $n > r$ take $c:=n/(n-r)$ and substitute 
\(
\frac{\sqrt{c}}{n} \coth^{-1}(\sqrt{c})
\)
for $\phi_n(r)$. Reasoning comes from Lemma \ref{l:arcsin}'s integral approximation to $\phi_n$:
in the limit as $n \to \infty, m \to \infty, \frac{n}{n-m} \to c$, 
\[
(n-m)\phi_n(m) \to
\int_0^1 \frac{1}{c - t^2} dt = 
\frac{\coth^{-1}(\sqrt{c})}{\sqrt{c}}.
\]
There is an exact formula for $\phi_n(1)$, which we omit because $m_{x, y}$ is at least $K$,
which in cases of interest will be at least 16.
\end{enumerate}

\subsection{The PaNNLD algorithm}
In the \texttt{PaLD} paradigm, we must compare entries in the matrix $(C_{x,v}^{F})_{\{x, v\} \in \mathcal{P}}$
with the cluster threshold $\tau$, which (\ref{e:clusterthreshold}) defines as
$(\nicefrac{1}{2 n})\sum_x C_{x,x}^{F}$. Here is the explicit formula, based on (\ref{e:randomcohesion}) and Theorem \ref{t:averages},
bearing in mind that half the
sum over $x$ and then over $y \in \mathcal{P}_x$ is the same as the sum over unordered pairs
$\{x, y\} \in \mathcal{P}$, and similarly for $\mathcal{R}$.

\begin{cor}\label{c:pannldthreshold}
The threshold (\ref{e:clusterthreshold}) for PaNNLD is a sum $\tau = \tau_{\mathcal{P}} + \tau_{\mathcal{R}}$
from promoted and relegated pairs, respectively:
\[
\tau_{\mathcal{P}}:= 
\frac{1}{n (n-1)} 
\sum_{\{x, y\} \in \mathcal{P}} \frac{1}{ |U^{\mathcal{P}}_{x, y}|}; \quad
\tau_{\mathcal{R}}:=
\frac{1}{n (n-1)} \sum_{r \geq 2} \phi_n(r) \cdot
| \{ \{x, y\} \in \mathcal{R}: \, m_{x, y} = r \} |.
\]
\end{cor}

\begin{defn}\label{d:pannld}
The \texttt{PaNNLD} unsupervised learning algorithm, consists of identifying the subset
$\mathcal{P}^*$ of promoted pairs $\{x, v\} \in \mathcal{P}$ for which $\min \{C_{x,v}^{F}, C_{v,x}^{F} \} \geq \tau$ from Corollary \ref{c:pannldthreshold}, and then partitioning the graph $(S, \mathcal{P}^*)$
into its connected components.
\end{defn}

For the reader's convenience, the ingredients of the \texttt{PaNNLD}
algorithm are summarized in Table \ref{t:pannld-summary}.

\begin{table}[]
    \centering
    \begin{tabular}{|c|c|c|} 
    \hline
    &  Promoted pairs $\mathcal{P}$ & Relegated pairs $\mathcal{R}$ \\
   & (\ref{e:promotedpairs})  &  (\ref{e:relegation}) \\
             \hline
  Useful data & Compare $\{x, y\}, \, \{x, z\} \in \mathcal{P}$& 
  $\{m_{x, y}, \{x, y\} \in \mathcal{R}\}$ histogram  \\
          & (\ref{e:direction}) & (\ref{e:rangeofinfluence})\\
                         \hline
\rule{0pt}{4ex}  Computational steps & $(\nicefrac{3}{2}) \sum\limits_{x \in S}  d_x^2$ & $2 \sum\limits_{x \in S}  d_x^2$
   \\
cf. Proposition \ref{p:cubic} &  Algorithm \ref{a:promocohesion} & Proposition \ref{p:averagingwork} \\
                         \hline
\rule{0pt}{4ex}   Conflict focus of a pair & $U^{\mathcal{P}}_{x, y}$ & $U^{\mathcal{R}}_{x, y}(\eta)$ \\
 cf. (\ref{e:conflictfocus})
 & Def. \ref{d:promoconflictfocus}, Algorithm \ref{a:promocohesion} & Definition \ref{d:relegconflictfocus}\\
    \hline
 \rule{0pt}{4ex}  Cohesion matrix $C^F$ (sum) & $C_{x,v}^{\mathcal{P}}$ & $\E[C_{x,v}^{\mathcal{R}}(\eta)]$\\
 cf. Definition \ref{d:cohesion} & Definition \ref{d:pannldmatrix} &
 Def. \ref{d:pannldmatrix}, Sect. \ref{a:intersection}, \ref{a:relegcohesion}\\
    \hline
  Cohesion threshold $\tau$ (sum) & $\tau_{\mathcal{P}}$ & $\tau_{\mathcal{R}}$ \\
  cf. (\ref{e:clusterthreshold}) & Corollary \ref{c:pannldthreshold} &
Corollary \ref{c:pannldthreshold}   \\

\hline
    \end{tabular}
    \medskip
    \caption{SUMMARY OF \texttt{PaNNLD}:
    \textit{
  Columns two and three summarize the contributions of promoted
  pairs and relegated pairs to the \texttt{PaNNLD} algorithm,
  with pointers to corresponding ingredients of \texttt{PaLD} in column one.
  Definition \ref{d:pannld} explains how to combine the last two rows to render a clustering of $S$.
    }}
    \label{t:pannld-summary}
\end{table}

\subsection{Concentration of measure}
Let us entertain a data modelling paradigm in which most of the information needed 
to classify object $x \in S$ is available by making comparisons among elements of $\mathcal{P}_x$,
i.e. among the in-neighbors and out-neighbors of $x$ in the nearest neighbor digraph.
In this paradigm, comparisons among other objects, e.g. those in $\mathcal{R}_x$, are treated as
background noise, to be replaced by averages. These averages appear in the second summand
of the formula (\ref{e:pannldmatrix}) for the PaNNLD cohesion matrix.

It is not enough for PaNNLD to reduce the computational complexity from $O(n^3)$ (the case
of PaLD) to the sum of squares of vertex degrees in the graph $(S, \mathcal{P})$, which
in favorable situations is $O(n K^2)$. Practitioners also need some assurance about the
quality of approximation when the \texttt{Comparator} at $x$ is invoked only for pairs 
$y, z \in \mathcal{P}_x$. 

We offer such an assurance founded on a concentration inequality. It rests on the premise
that statistical properties of the true relations $\{y \preceq_x z, \{y, z\} \subset \mathcal{R}_x\}$
are comparable to those of the relations 
$\{\eta_{x,y} <  \eta_{x,z}, \{y, z\} \subset \mathcal{R}_x\}$ from the stochastic model 
(\ref{e:randomizerelegated}). This premise is consistent with our initial paradigm:
outside the set of in-neighbors and out-neighbors of $x$, in the nearest neighbor digraph,
the ``friend-of-a-friend principle'' \cite{bar} is likely to break down\footnote{
For any distinct $x, y, z$, small values of $\eta_{x,y}$ and $\eta_{y,z}$ do not
imply a small value for $\eta_{x,z}$, because these three random variables are independent.
}, and so
triplet comparisons can be treated as random.

Theorem \ref{t:cohesionconcentration} compares the stochastic matrix $(C_{x,v}^{\mathcal{R}}(\eta))$
of (\ref{e:randomcohesion}), based on the stranger randomization (\ref{e:randomizerelegated}),
with the average which appears in (\ref{e:pannldmatrix}); similarly for the trace.

\begin{thm}[CONCENTRATION OF MEASURE]\label{t:cohesionconcentration}
Let $K:=\min\{K_x, x \in S\}$ be the minimum friend set size. Then for any $\theta > 0$
\begin{enumerate}
    \item[(i)]
    \(
\p[ \left| C^{\mathcal{R}}_{x, v}(\eta) -\E[C^{\mathcal{R}}_{x, v}(\eta)] \right| \geq \theta ]
\leq 2 e^{-\theta^2 K^2}
\) if $ \{x, v\} \in \mathcal{P}$, or $x = v$.
    \item[(ii)] The trace of $((\nicefrac{1}{2 n})C_{x,v}^{\mathcal{R}}(\eta) )$ is concentrated around the parameter $ \tau_{\mathcal{R}}$ in Corollary \ref{c:pannldthreshold}:
    \[
\p[ | \frac{1}{2n}\sum_x C^{\mathcal{R}}_{x, x}(\eta) - \tau_{\mathcal{R}} | \geq \frac{\theta}{n} ] \leq 2e^{-(2 \theta K /3)^2}.
\]
\end{enumerate}
\end{thm}

The practical implication of (i) is that choosing $K$ significantly bigger
than $1/\theta$ controls the randomization-induced error in \texttt{PaNNLD}, uniformly in $n$.
Proof of Theorem \ref{t:cohesionconcentration} will be given in Section \ref{s:concentration}.

\section{Complexity of computing the promoted cohesion matrix}\label{s:promocomplexity}
 
\subsection{Pre-computation of ranking tables for each $x \in S$}\label{s:lookuptable}
Denote $|\mathcal{P}_x|$ by $d_x$.
For each $x \in S$, we begin by sorting $\mathcal{P}_x$, with respect to the \texttt{Comparator}
at $x$, and placing the
results in a Table $T_x$, which maps each $y \in \mathcal{P}_x$ to the
integer from $1$ to $d_x$ denoting its rank.
The \texttt{sort} method of Java \texttt{List} \cite{javalist}, for example,
accomplishes this with 
\(
\sum_{x \in S}  d_x \log{d_x} 
\)
oracle calls; by inequality (\ref{e:promographdegrees}) the latter is 
bounded above by $2 n \bar{K} \log{n}$ .
When a later decision is needed about whether $y \prec_x z$,
for $y, z \in \mathcal{P}_x$, compare the rank of $y$ with the rank of $z$ in $T_x$.

\subsection{Graph traversal for cohesion matrix calculations}
We describe below an efficient vertex traversal of $(S, \mathcal{P})$, 
which returns the sets
$\{U^{\mathcal{P}}_{x \| y}, \{x, y\} \in \mathcal{P}\}$ of Definition \ref{d:leftpromoconflictfocus}, 
and deterministic ingredients of each of the sets in
Definition \ref{d:leftrelegconflictfocus}, namely the \textbf{connector sets}
 \begin{equation}\label{e:connectors}
     D^{\mathcal{R}}_{y \| z} := \{x \in \mathcal{P}_y \cap \mathcal{P}_z, \, y \prec_x z\}, \quad  \{y, z\} \in \mathcal{R}.
 \end{equation}
Evidently $\mathcal{P}_y \cap \mathcal{P}_z$ is the disjoint union 
\(
    D^{\mathcal{R}}_{y \| z} \cup D^{\mathcal{R}}_{z \| y}.
\)
Given (\ref{e:connectors}) and the vertex degrees in $(S, \mathcal{P})$,
the range of influence (\ref{e:rangeofinfluence}) is expressible as:
\[
m_{x,y}:=|\{x, y\} \cup (\mathcal{P}_x \cup \mathcal{P}_y)| =
2 + d_x + d_y - |\mathcal{P}_x \cap \mathcal{P}_y| = 
2 + d_x + d_y - |D^{\mathcal{R}}_{y \| x}| - |D^{\mathcal{R}}_{x \| y}|.
\]
All of this will be achieved within a budget of $\sum_x \binom{|\mathcal{P}_x|}{2}$ triplet comparisons,
a function of the vertex degrees in the undirected graph $(S, \mathcal{P})$.

No specific ordering of $(S, \mathcal{P})$ is required, and in Java the computation can
be spread across processors using a parallel stream.

\subsection{Efficient algorithm for PaNNLD}

\begin{algorithm}
  \caption{PaNNLD promoted cohesion matrix}\label{a:promocohesion}
  \begin{algorithmic}[1]
    \Procedure{Half-focus membership}{$\{\prec_x, x \in S\}$, $\mathcal{P}$}
      \State $U^{\mathcal{P}}_{x \| y}:= \{x\}, \, \forall \{x, y\} \in \mathcal{P}$
      \Comment{Half-focus initialization}
      \State $D^{\mathcal{R}}_{y \| z} := \emptyset, \, \forall \{y, z\} \in \mathcal{R}$ (lazy initialization)
      \Comment{Connector initialization}
 \For{$x \in S$}
      \For{$ y, z \in \mathcal{P}_x$ (distinct, unordered)}
      \Comment{$\sum_x \binom{|\mathcal{P}_x|}{2}$ choices of $x, y, z$}
      \If{$\{y, z\}$ is relegated}
            \If{$y \prec_x z$}
             \State add $x$ to $D^{\mathcal{R}}_{y \| z}$
                \State add $y$ to $U^{\mathcal{P}}_{x \| z}$
                \Comment{``promoted beats relegated''}
            \Else{ $z \prec_x y$}
            \State add $x$ to $D^{\mathcal{R}}_{z \| y}$
            \State add $z$ to $U^{\mathcal{P}}_{x \| y}$
            \Comment{``promoted beats relegated''}
            \EndIf
        \Else{ $\{y, z\}$ is promoted}
        \\
\Comment {edges $\{x, y\}$, $\{y, z\}$, and $\{x, z\}$ form a triangle in $(S, \mathcal{P})$}
      \If{$x \prec_y z$ and $y \prec_x z$}
            \State add $x$ to $U^{\mathcal{P}}_{y \| z}$
            \Comment{follow Definition \ref{d:leftpromoconflictfocus}}
        \EndIf
      \If{$x \prec_z y$ and $z \prec_x y$}
            \State add $x$ to $U^{\mathcal{P}}_{z \| y}$
            \Comment{follow Definition \ref{d:leftpromoconflictfocus}}
        \EndIf
    \EndIf
      \EndFor
\EndFor
      \State \textbf{return: }$\{U^{\mathcal{P}}_{x \| y}, \{x, y\} \in \mathcal{P} \}$
            \Comment{$U^{\mathcal{P}}_{x, y} = U^{\mathcal{P}}_{x \| y} \cup U^{\mathcal{P}}_{y \| x}$}
      \State \textbf{return: }non-empty sets among $\{D^{\mathcal{R}}_{y \| z}, \{x, y\} \in \mathcal{R} \}$
      \Comment{$\mathcal{P}_y \cap \mathcal{P}_z=
    D^{\mathcal{R}}_{y \| z} \cup D^{\mathcal{R}}_{z \| y}$}
    \EndProcedure
  \end{algorithmic}
\end{algorithm}

The following algorithm, central to PaNNLD, is presented twice: first as Algorithm \ref{a:promocohesion},
and secondly in text.
Proof of Theorem \ref{t:complexity}
amounts to checking the number of triplet comparisons in Algorithm \ref{a:promocohesion}, and the sufficiency of its output.

For $\{x, y\} \in \mathcal{P}$ initialize $U^{\mathcal{P}}_{x \| y}$ to equal the singleton
$\{x\}$. Lazily initialize $D^{\mathcal{R}}_{y \| z}$ to be empty,
which means the set is not created at all unless some element shows up.
For each $x \in S$, do the following for every pair of distinct neighbors of $x$ in $\mathcal{P}$,
i.e. for every $y, z \in \mathcal{P}_x$.

\begin{enumerate}
    \item[(R)] 
    If $\{y, z\}$ is relegated, then 
    \begin{enumerate}
        \item[(R1)]
        If $y \prec_x z$,  add $x$ to $D^{\mathcal{R}}_{y \| z}$, and add 
        $y$ to $U^{\mathcal{P}}_{x \| z}$; ``promoted beats relegated''.
        
        \item[(R2)]
        Else $z \prec_x y$; add $x$ to $D^{\mathcal{R}}_{z \| y}$, and add $z$ to $U^{\mathcal{P}}_{x \| y}$; ``promoted beats relegated''.

    \end{enumerate}

    \item[(P)] 
    If $\{y, z\}$ is promoted, then edges $\{x, y\}$, $\{y, z\}$, and $\{x, z\}$ form a triangle in $(S, \mathcal{P})$,
    and we place $x$ as Definition \ref{d:leftpromoconflictfocus} dictates\footnote{
    Postpone the placements of $y$ and $z$ until their turn comes along.}:
    \begin{enumerate}
          \item[(P1)] 
     If $x \prec_y z$ and $y \prec_x z$, add $x$ to $U^{\mathcal{P}}_{y \| z}$.  
               \item[(P2)] 
     If $x \prec_z y$ and $z \prec_x y$, add $x$ to $U^{\mathcal{P}}_{z \| y}$. 
 
    \end{enumerate}

\end{enumerate}

\subsection{Proof of Theorem \ref{t:complexity}}
The Tables built in Section \ref{s:lookuptable} establish statement (a) of Theorem \ref{t:complexity}.
Algorithm \ref{a:promocohesion} has an outer loop over
$x \in S$, and an inner loop over $\binom{d_x}{2}$ 
pairs of distinct neighbors of $x$ in $\mathcal{P}$. Inspection of options R1 and R2
shows that only one triplet comparison is needed in the case of a relegated pair.
P1 and P2 show that three triplet comparisons are needed whenever
edges $\{x, y\}$, $\{y, z\}$, and $\{x, z\}$ form a triangle in 
$(S, \mathcal{P})$. This means look-ups into the each of the precomputed sorted sets
$(\mathcal{P}_x, \prec_x)$, $(\mathcal{P}_y, \prec_y)$, and $(\mathcal{P}_z, \prec_z)$
Thus the total budget is no more than \(
3 \sum_x \binom{d_x}{2}
\) operations, which proves statement (b) of Theorem \ref{t:complexity}.
Algorithm \ref{a:promocohesion} returns the sets 
$\{U^{\mathcal{P}}_{x \| y}, \{x, y\} \in \mathcal{P} \}$, from which we
recover the numerators, and also the denominators
\(
|U^{\mathcal{P}}_{x, y}| = |U^{\mathcal{P}}_{x \| y}| + |U^{\mathcal{P}}_{y \| x}|,
\)
of the expression for the matrix $C^{\mathcal{P}}_{x, v}$
(Definition \ref{d:pannldmatrix}).
This completes the proof of Theorem \ref{t:complexity}.

\subsection{Pathological example: path metric on a star graph}\label{s:stargraph}

Take real numbers $0 < w_1 < w_2 < \ldots < w_n$. Consider a star graph centered at vertex $x_0$,
with leaf vertices $x_1, x_2, \ldots, x_n$, meaning that the edge set $E$ is $\{ \{x_0, x_j\}, 1 \leq j \leq n \}$.
Let $S:= \{x_0, x_1, x_2, \ldots, x_n \}$ denote the vertex set. Place a metric $\rho$ on $S$, corresponding
to weighted path length in the star graph $(S, E)$, where edge $ \{x_0, x_j\}$ has weight $w_j$.
In other words, $\rho(x_0, x_j) = w_j$, and
\[
\rho(x_i, x_j) = w_i + w_j, \quad 1 \leq i < j \leq n.
\]
The corresponding ranking system on $S$ (which is the same for any strictly increasing positive
sequence $(w_j)$) has
\[
x_i \prec_{x_k} x_j \quad 0 \leq i < j \leq n, \quad k \notin \{i, j\}.
\]
Experiments show that \texttt{PaLD} clustering places about $62\%$ of the vertices (those closest to $x_0$) in
a single cluster, while the remaining vertices are singleton
outliers.

Suppose friend sets are given by two nearest neighbors.
Then friend sets may be catalogued as  $x_0  \leftarrow  x_j \rightarrow x_1$
for integers $j = 2, \ldots, n$, together with
\[
x_1  \leftarrow  x_0 \rightarrow x_2; \quad
x_0  \leftarrow  x_1 \rightarrow x_2.
\]
Thus $x_0$ has in-degree $n$. The $2 n - 1$ promoted pairs are
\[  
\mathcal{P}= \{ \{x_0, x_i\}, \, 1 \leq i \leq n \} \cup \{ \{x_1, x_j\}, \, 2 \leq j \leq n \}.
\]
The relegated pairs are
\(
\mathcal{P}= \{ \{x_i, x_j\}, \, 2 \leq i < j \leq n \}.
\)
Algorithm \ref{a:promocohesion} is poorly suited to this case, since vertices $x_0$ and $x_1$
both have degree $n$, giving roughly $n^2$ triplet comparisons.
Thus a budget of $O(n K^2)$ computational steps in insufficient in such a case,
at least using \texttt{PaNNLD}.

\section{Averaging over the stranger randomization}\label{s:averages}

\subsection{Techniques for stranger averages}
Our goal in Section \ref{s:averages} is a constructive proof of Theorem \ref{t:averages}.
The two main ideas are:
\begin{enumerate}

    \item Grouping relegated pairs $\{x, y\}$ for which $m_{x, y}$ has a specific value, to avoid the $O(n^2)$ cost of summation over all $\{x, y\} \in \mathcal{R}$.
    
    \item Expressing the mean of $1/|U^{\mathcal{R}}_{x, y}|$ when 
    $\{x, y\}$ is a relegated pair, in terms of $m_{x, y}:= |\{x, y\} \cup (\mathcal{P}_x \cup \mathcal{P}_y)|$.

\end{enumerate}

Begin with (1), and exploit output of Algorithm \ref{a:promocohesion} to provide data
needed for (2).

\subsection{Probabilistic result for stranger averages}\label{s:randomrelegates}

Recall Definition \ref{d:relegconflictfocus}:
\[
    U^{\mathcal{R}}_{x, y}:= 
    \{x, y\} \cup \mathcal{P}_x \cup \mathcal{P}_y \cup
    \{z \in \mathcal{R}_x \cap \mathcal{R}_y: \eta_{x,y} > \min \{  \eta_{x,z}, \eta_{y,z} \} \}.
\]
By the independence and distribution assumptions of (\ref{e:randomizerelegated}), the size of the
random set $U^{\mathcal{R}}_{x, y}$ has the following distribution:

\begin{lem}\label{l:binomialsize}
Suppose the relegated pair $\{x, y\}$ satisfies $m_{x, y}:= |\{x, y\} \cup \mathcal{P}_x \cup \mathcal{P}_y| \leq n-1$.
Conditional on $\eta_{x,y} = 1-t \in (0, 1)$, $|U^{\mathcal{R}}_{x, y}| - m_{x, y}$ has a
Binomial$(n - m_{x, y} , 1 - t^2)$ distribution.
\end{lem}

\begin{proof}
The set $\mathcal{R}_x \cap \mathcal{R}_y$ is the complement of $\{x, y\} \cup \mathcal{P}_x \cup \mathcal{P}_y$,
and has size $n - m_{x, y}$, which is at least 1 by assumption.
Given $\eta_{x,y} = 1-t$, the probability of the event 
$\min \{  \eta_{x,z}, \eta_{y,z} \} \geq \eta_{x,y} = 1-t$ is $t^2$,
for $z \in \mathcal{R}_x \cap \mathcal{R}_y$, and the complementary event
$\min \{  \eta_{x,z}, \eta_{y,z} \} < \eta_{x,y}$ has conditional probability $1 - t^2$.
These events are independent as $z$ varies over $\mathcal{R}_x \cap \mathcal{R}_y$, so 
\[
|\{z \in \mathcal{R}_x \cap \mathcal{R}_y: \eta_{x,y} > \min \{  \eta_{x,z}, \eta_{y,z} \} \}|
\sim \mbox{Binomial}(n - m_{x, y} , 1 - t^2).
\]
Add this to $m_{x, y}$ to obtain $|U^{\mathcal{R}}_{x, y}|$, which gives the result.
\end{proof}

For any $\{x, y\} \in \mathcal{R}$, $\eta_{x,y}$ is a Uniform$(0,1)$ random variable.
It follows from Lemma \ref{l:binomialsize} that the mean of $|U^{\mathcal{R}}_{x, y}|^{-1}$
depends only on $n$ and the range of influence $m_{x, y}$ from (\ref{e:rangeofinfluence}), 
for any $\{x, y\} \in \mathcal{R}$. This establishes the first assertion of
Theorem \ref{t:averages}. There is a function $\phi_n$, whose
analytic approximation is given in Lemma \ref{l:arcsin}, so that
\begin{equation}\label{e:inversemoment}
    \E \left( \frac{1}{|U^{\mathcal{R}}_{x, y}| } \right)
    = \phi_n(d_x + d_y + 2 - |\mathcal{P}_x \cap \mathcal{P}_y|).
    \end{equation}
Validity of Lemma \ref{l:arcsin}, for large $n$ and $K$, depends on the fact that
the range of influence $m_{x, y}$, as in (\ref{e:rangeofinfluence}),
satisfies 
\begin{equation}\label{e:lowerboundpromo}
  m_{x, y} \geq 2 + \max \{ K_x, K_y \} \geq K+2, \quad \forall \{x, y\} \in \mathcal{R}.
\end{equation}
Lemma \ref{l:arcsin} shows that for $m_{x, y} < n$
\[
(n-m_{x, y}) \phi_n(m_{x, y}) = 
\E \left( \frac{n-m_{x, y}}{|U^{\mathcal{R}}_{x, y}| } \right)  = 
\int_0^1 \frac{1}{c_{x, y} - t^2} dt +\epsilon_{n, K}
\]
where 
\(
c_{x, y}:=n/(n - m_{x, y}),
\)
and $\epsilon_{n, K} \to 0$ as $n, K \to \infty$. Otherwise $\mathcal{R}_x \cap \mathcal{R}_y = \emptyset$,
in which case $|U^{\mathcal{R}}_{x, y}|= n$, and $\phi_n(n) = 1/n$.

\subsection{Grouping vertices with the same degree}
Let $\Lambda_1$ denote the set of distinct values of $\{d_x + 1, x \in S\}$, and let
$\Lambda_2$ denote the set of distinct unordered pairs $\{\alpha, \beta\}$
for which $d_x = \alpha-1$, $d_y=\beta-1$, for some pair $x, y$ in $S$.
Abbreviate the mean out-degree of the nearest neighbors digraph $(S, F)$ to
$\bar{K}:=\sum_x K_x / n$. Proposition \ref{p:boundvertexdegreepairs} places an
$O(n \bar{K})$ cap on $|\Lambda_2|$.

\begin{prop}\label{p:boundvertexdegreepairs}
In the undirected graph $(S, \mathcal{P})$, with $n \bar{K}$ edges,
and minimum vertex degree at least $K$, take  $\lambda:= \sqrt{2n(2 \bar{K} -  K)} + 1$. Then
\[
|\Lambda_1| < \lambda; \quad
|\Lambda_2| < \lambda(\lambda+1)/2.
\]
Thus a table of values $\{\phi_n(\alpha+\beta): \{\alpha, \beta\} \in \Lambda_2\}$
can be obtained with $O(n \bar{K})$ work.
\end{prop}
Proposition \ref{p:boundvertexdegreepairs}
follows from applying Lemma \ref{l:distinctdegrees} to the graph 
$(S, \mathcal{P})$.

\begin{lem}\label{l:distinctdegrees}
Let $(V, E)$ be an undirected graph on $n$ vertices with $m$ edges, and minimum vertex degree at least $K$. 
The number $t$ of distinct vertex degrees in $(V, E)$ is less than
$\sqrt{(4m - 2 n K)} + 1$.
\end{lem}

\begin{proof}
Let the vertex degrees of $(V, E)$ be
\(
d_1 \geq d_2 \geq \cdots \geq d_n \geq K.
\)
Given that there are $t$ distinct vertex degrees, there must exist indices
\(i(1) < i(2) < \cdots < i(t-1)
\)
for which
\(
d_{i(j)}\geq 1 + d_{i(j+1)}
\). Hence
\[
2 m = \sum_1^n d_i \geq n K + 1 + 2 + \cdots + t-1 = \frac{(t-1)t}{2} + nK.
\]
This implies $(t-1)^2 < 4m - 2 n K$, which gives the result.
\end{proof}

\subsection{Partial sums}\label{s:partialsums}
We shall use (\ref{e:inversemoment}) to simplify computation of the mean
of (\ref{e:randomdiagonal}), namely
\begin{equation}\label{e:meandiagonal}
  G_x:=\sum_{y \in \mathcal{R}_x} \phi_n(2+d_x+ d_y- |\mathcal{P}_x \cap \mathcal{P}_y |)
= \sum_{y \in \mathcal{R}_x } 
\E \left( \frac{1}{|U^{\mathcal{R}}_{x, y}(\eta)| } \right)
= (n-1) \E[C^{\mathcal{R}}_{x, x}(\eta)].  
\end{equation}
Modification to this expression will also supply $\E[C^{\mathcal{R}}_{x, v}(\eta)]$ appearing in
Theorem \ref{t:averages}. Naive computation of $\{G_x, x  \in S\}$ by 
(\ref{e:meandiagonal}) would require a sum over
$\mathcal{R}_x$ for every $x \in S$, forcing $O(n^2)$ work. To avoid this, Definition \ref{d:sumsreciprocals}
provides approximations to $G_x$ when $d_x = \alpha-1$.

\begin{defn}[DIAGONAL APPROXIMATION]\label{d:sumsreciprocals}
For $\alpha \geq 1$, define 
\[
g(\alpha):=\sum_{\beta = 1}^{n - \alpha} \phi_n(\alpha + \beta) \cdot|  \{y: d_y = \beta-1\} |
\]
and for each $x \in S$, define a diagonal approximation (for the relegated cohesion):
\begin{equation}\label{e:sumsmissingintersections}
  H_x:=g(d_x + 1) - \sum_{
  \substack{y \in \mathcal{P}_x \cup \{x\} \\ d_y + d_x \leq n - 2}} \phi_n(d_x + d_y+2)
  = \sum_{y \in \mathcal{R}_x} \phi_n(d_x+ d_y+2).
\end{equation}
\end{defn}

Direct calculation shows that
the difference between (\ref{e:meandiagonal}) and (\ref{e:sumsmissingintersections}) is:
\begin{equation}\label{e:intersectiondifference}
  G_x - H_x =\sum_{\substack{y \in \mathcal{R}_x\\ \mathcal{P}_x \cap \mathcal{P}_y \neq \emptyset}}
(\phi_n(m_{x, y}) - \phi_n(d_x + d_y+2)).  
\end{equation}

To obtain the expected value of (\ref{e:randompromoted}), we will modify $G_x$
in (\ref{e:meandiagonal}) to give
\begin{equation}\label{e:meanoffdiagonal}
    G_{x, v}:=G_x -\sum_{\substack{y \in \mathcal{P}_v \cap \mathcal{R}_x\\
     y \prec_v x} }\phi_n(m_{x, y}), \quad \{x, v\} \in \mathcal{P}.
\end{equation}

\begin{prop}\label{p:averagingwork}
The values $\{H_x, x  \in S\}$  are computable in $O(n \bar{K})$ work.
The intersection algorithm \ref{a:intersection} below obtains
$\{G_x, x  \in S\}$ through a further $\sum_x d_x^2$ steps.
The relegated cohesion matrix algorithm \ref{a:relegcohesion} below obtains
$\{G_{x, v}, \{x, v\} \in \mathcal{P} \}$ in $\sum_x d_x^2$ more steps.
\end{prop}

\begin{proof}
Proposition \ref{p:boundvertexdegreepairs} showed that
the table of values of $\phi_n(\alpha + \beta)$ can be obtained 
with $O(n \bar{K})$ work, as can the counts  $| \{y: m_y = \beta\} |$ for each $\beta$,
of which there are at most $2 \sqrt{n \bar{K}} + 1$ choices.
Compute each $g(\alpha)$ by summing at most $2 \sqrt{n \bar{K}} + 1$ terms, 
for at most $2 \sqrt{n \bar{K}} + 1$ choices of $\alpha$.
This shows how all $(g(\alpha))$ are computable in $O(n \bar{K})$ work.

For each $x$, a further $d_x$ subtractions provides the value of $H_x$, using (\ref{e:sumsmissingintersections}).
Thus all values $\{H_x, x  \in S\}$  are obtainable in $\sum_x d_x = 2 n \bar{K}$ steps.

Formula (\ref{e:intersectiondifference}) shows how the values $\{G_x, x  \in S\}$
can be deduced from $\{H_x, x  \in S\}$, via
\[
| \{\{x, y\} \in \mathcal{R}: \mathcal{P}_x \cap \mathcal{P}_y \neq \emptyset \} | \leq \sum_x d_x^2
\]
further modifications. As for (\ref{e:meanoffdiagonal}), there are $d_x$ choices of $v$ for each $x$,
and for each $v$ a further $d_x$ terms must be added to obtain $G_{x, v}$, giving $\sum_x d_x^2$
steps to derive the sparse matrix $\{G_{x, v}, \{x, v\} \in \mathcal{P} \}$.
\end{proof}

\subsection{Algorithms}

\subsubsection{Intersection algorithm} \label{a:intersection}
\textbf{Goal: }The first algorithm starts by computing each $H_x$ (Definition \ref{d:sumsreciprocals}),
and will derive each $G_x$ using the difference (\ref{e:intersectiondifference}). 
\begin{enumerate}
    \item[(0)]
    For each $x \in S$, initialize $G_x:=H_x$.
    
    \item[(1)] 
      Iterate through the (at most $\sum_x |\mathcal{P}_x|^2$) relegated pairs 
      $\{x, y\}$ for which
$D^{\mathcal{R}}_{x, y}$ is non-empty, as constructed in Algorithm \ref{a:promocohesion}. 
At such $\{x, y\}$, use the notations:
\[
\alpha = d_x + 1; \quad \beta = d_y + 1; \quad \gamma = |D^{\mathcal{R}}_{x, y}| = 
|\mathcal{P}_x \cap \mathcal{P}_y|; \quad 
\delta:= \phi_n(\alpha + \beta- \gamma) - \phi_n(\alpha + \beta).
\]
\begin{enumerate}
    \item[(a)]
Add \(\delta\) to $G_{x}$.
    \item[(b)]
Add \(\delta\) to $G_{y}$.
\end{enumerate}
\item[(2)]
    Return $\{G_x: x \in S \} $.
\end{enumerate}

\subsubsection{Relegated cohesion matrix algorithm} \label{a:relegcohesion}
\textbf{Goal: }The second algorithm will compute the mean of off-diagonal terms (\ref{e:randompromoted})
of the relegated cohesion matrix (Definition \ref{d:pannldmatrix}).
Indeed $G_{x, v} = (n-1) \E[C^{\mathcal{R}}_{x, v}(\eta)]$ will result from
modifying $G_x$ and $G_v$ for each promoted pair $\{x, v\}$.

\begin{enumerate}
\item[(0)] Initialize $G_{x, v}:=G_x$ and $G_{v, x}:=G_v$.
\item[(1)] Iterate over $y \in  \mathcal{P}_v \triangle \mathcal{P}_x$.
\begin{enumerate}
    \item[(A)]
     If $y \in \mathcal{P}_v \cap \mathcal{R}_x$ and $y \prec_v x$,
    then subtract $\phi_n(m_{x, y})$ from $G_{x, v}$, as (\ref{e:randompromoted}) suggests.
    \item[(B)] If $y \in \mathcal{P}_x \cap \mathcal{R}_v$ and $y \prec_x v$,
    then subtract $\phi_n(m_{v, y})$ from $G_{v, x}$ (i.e. roles of $x$ and $v$ are reversed).
\end{enumerate}
  
 \item[(2)] Return $\{G_{x, v}: \{x, v\} \in\mathcal{P} \} $.
\end{enumerate}

\subsection{Proof of Theorem \ref{t:averages}}
Let us summarize the results of this section, which prove Theorem \ref{t:averages}.
The algorithms \ref{a:intersection} and \ref{a:relegcohesion}
produce values for $\{G_{x}: x \in S\}$ and $\{G_{x, v}:  \{x, v\} \in\mathcal{P} \}$, which are linked by
formulas (\ref{e:meandiagonal}) and (\ref{e:intersectiondifference}) 
to the averages presented in Theorem \ref{t:averages}, i.e
\begin{equation}\label{e:relegachieved}
\E[C^{\mathcal{R}}_{x, x}(\eta)] = \frac{G_{x}}{n-1}; \quad
\E[C^{\mathcal{R}}_{x, v}(\eta)] = \frac{G_{x, v}}{n-1}, \quad \{x, v\} \in\mathcal{P}.
\end{equation}
The connectors sets $(D^{\mathcal{R}}_{y \| z})$ (see formula (\ref{e:connectors})),
which were derived in Algorithm \ref{a:promocohesion} with no more than 
$\sum_x d_x^2$ triplet comparisons,
were used in algorithm \ref{a:intersection}.
Proposition \ref{p:averagingwork} shows that
algorithms \ref{a:intersection} and \ref{a:relegcohesion}
can be completed with $O(\sum_x d_x^2)$ work.

\section{Proving concentration of measure for stranger randomizations}\label{s:concentration}

The goal of this Section is to prove Theorem \ref{t:cohesionconcentration}.

\subsection{Mechanics of the Bounded Differences Inequality}
Suppose $\xi:=(\xi_i)_{1 \leq i \leq m}$ is a collection of independent random variables,
and $F(.)$ is a real function of $m$ real variables
.
Suppose that for each index $t \leq m$ there is an upper bound $c_t$
on the difference in $F$ caused by variation in the $t$-th coordinate. In other words,
for any choice of $x_1, \ldots, x_t, \ldots, x_m$, and any alternative value
$x'_t$ for the $t$-th coordinate,
\begin{equation}\label{e:differencebound}
  |F(x_1, \ldots, x_{t-1}, x'_t, x_{t+1}, \ldots, x_m) - F(x_1, \ldots, x_{t-1}, x_t, x_{t+1}, \ldots, x_m)| \leq c_t.  
\end{equation}
The collective effect of all the $(c_t)$ is summarized in the parameter
\[
\nu:=\frac{1}{4} \sum_{t=1}^m c^2_t.
\]
The well known Bounded Differences Inequality, 
as stated by Boucheron, Lugosi, Massart \cite[Theorem 6.2]{bou}  says that
\begin{equation}\label{e:bde}
    \p[F(\xi) - \E[F(\xi)] \geq t ] \leq e^{-t^2/(2 \nu)}.
\end{equation}
The work of Section \ref{s:concentration} consists in selecting suitable functions $F$, and computing $\nu$.

\subsection{Application of bounded differences to stranger randomization}
The argument of our function $F$ will be the collection of independent random variables
$\eta:=\{\eta_{z, w}, \{z, w\} \in \mathcal{R}\}$, introduced in Section \ref{s:strangerrand}
for purposes of stranger randomization.

Recall the relegated conflict focus $U^{\mathcal{R}}_{x, y}(\eta)$ (Definition \ref{d:relegconflictfocus}), whose
size is the sum of the range of influence $m_{x, y}$ (see (\ref{e:rangeofinfluence})) and
a random variable $T_{x, y}(\eta)$ depending on the stranger randomization:
\begin{equation}\label{e:strangercounts}
  T_{x, y}(\eta):=|U^{\mathcal{R}}_{x, y}(\eta)| - m_{x, y} = 
|\{z \in \mathcal{R}_x \cap \mathcal{R}_y: \eta_{x,y} > \min \{  \eta_{x,z}, \eta_{y,z} \} \}|.
\end{equation}
The first choice of $F$ is:
\begin{equation}\label{e:totalstrangerreciprocal}
   F(\eta) :=\sum_{ \{x, y\} \in \mathcal{R} } \frac{1}{|U^{\mathcal{R}}_{x, y}(\eta)|}
   = \sum_{ \{x, y\} \in \mathcal{R} } \frac{1}{ m_{x, y} + T_{x, y}(\eta)}.
\end{equation}
Our main task is to establish a bound (e.g. $c_t$ as in (\ref{e:differencebound})) on the difference in $F$ 
caused by variation in one co-ordinate; in this case, a co-ordinate takes the form $\{z, w\} \in \mathcal{R}$.

\begin{lem}[Difference Bounds]\label{l:estimatedifferences}
Fix an arbitrary $\{z, w\} \in \mathcal{R}$, and
let $\tilde{\eta}$ be the modification of $\eta$
in which $\eta_{z,w}$ is replaced by $\eta'_{z,w} \in (0, 1)$.
The variation in the function (\ref{e:totalstrangerreciprocal}) has the following upper bound:
\(  
|F(\tilde{\eta}) - F(\eta)| \leq c_{z, w},
\)
where 
\(
c_{z, w}:= \nicefrac{1}{m_{z, w}} + \nicefrac{1}{d_z} +\nicefrac{1}{d_w}
\), $d_z:=|\mathcal{P}_z|$, and  $m_{z, w} =|\{z, w\} \cup \mathcal{P}_z \cup \mathcal{P}_w|$
as in (\ref{e:rangeofinfluence}).
\end{lem}

\begin{proof}
Fix $\{z, w\} \in \mathcal{R}$, and let $\tilde{\eta}$ be the modification of $\eta$
described in the Lemma. For any $\{x, y\} \in \mathcal{R}$, let
$\delta^{z, w}_{x, y}$ denote the following difference:
\begin{equation}\label{e:sumdifferences}
    \delta^{z, w}_{x, y}:=  \left|
\frac{1}{m_{x, y} + T_{x,y}(\tilde{\eta})} - 
\frac{1}{m_{x, y} + T_{x,y}(\eta)}\right|.
\end{equation}
Consider $\{z, w\}$ and $\{x,y\}$ as edges in the graph $(S, \mathcal{R})$.
Inspection of (\ref{e:strangercounts}) shows $T_{x,y}(\tilde{\eta})=T_{x,y}(\eta)$
unless the edges $\{z, w\}$ and $\{x,y\}$  either coincide, or else have a common end point.
It follows that $|F(\tilde{\eta}) - F(\eta)|$ is bounded above by
\begin{equation}\label{e:differencebreakdown}
   \delta^{z, w}_{z, w} +  \sum_{y \in \mathcal{R}_z \setminus \{w\}} \delta^{z, w}_{z, y} + 
    \sum_{y \in \mathcal{R}_w \setminus \{z\}}  \delta^{z, w}_{w, y}. 
\end{equation}
In the first sum, $z$ is the common end point; in the second sum, $w$ is.
We have avoided counting $\delta^{z, w}_{z, w}$ twice.
Begin by bounding the first term in (\ref{e:differencebreakdown}). After that, we bound
the second term, and apply the same bound (switching $w$ and $z$) to the third term.

\noindent \textbf{Bound on $\delta^{z, w}_{z, w}$: }
The absolute difference between $T_{z, w}(\tilde{\eta})$ and
$T_{z,w}(\eta)$ could be anything from zero to $n-m_{z, w}$, so
\[
\delta^{z, w}_{z, w} \leq \frac{1}{m_{z, w}}.
\]

\noindent \textbf{Summing $\delta^{z, w}_{z, y}$ over $y \in \mathcal{R}_z \setminus \{w\}$: } 
Rank the elements $y(1), y(2), \ldots, y(m)$ of $\mathcal{R}_z$ so
\[
\eta_{z,y(1)} < \cdots < \eta_{z,y(m)}.
\]
There is some $\ell$ for which $w:=y(\ell)$. Our goal is to show that 
\(
 \sum_{j= 1}^{m}  \delta^{z, y(\ell)}_{z, y(j)}  - \delta^{z, y(\ell)}_{z, y(\ell)}
\)
cannot exceed $\nicefrac{1}{d_z}$, where $d_z:=|\mathcal{P}_z|$.  The first step will be to show that
\begin{equation}\label{e:increasingcf}
    |U^{\mathcal{R}}_{z, y(j)}(\eta)| \geq d_z + j,
\quad j = 1, 2, \ldots, m.
\end{equation}
Given $1 \leq i < j \leq m$, there are two possibilities:
either $\{y(i), y(j)\} \in \mathcal{P}$, which implies $y(i) \in U^{\mathcal{R}}_{z, y(j)}(\eta)$
by Definition \ref{d:relegconflictfocus}, or else $\{y(i), y(j)\} \in \mathcal{R}$, which implies
$y(i) \in  \mathcal{R}_z \cap \mathcal{R}_{y(j)}$, and so $y(i) \in U^{\mathcal{R}}_{z, y(j)}(\eta)$
because it contributes to the total (\ref{e:strangercounts}):
\[
T_{z, y(j)}(\eta):=
|\{u \in \mathcal{R}_z \cap \mathcal{R}_{y(j)}: \eta_{z,y(j)} > \min \{  \eta_{z,u}, \eta_{y(j),u} \} \}|,
\]
This analysis confirms (\ref{e:increasingcf}), because it shows that:
\[
\{z\} \cup \mathcal{P}_z \cup \{y(1), y(2), \ldots, y(j-1)\}  \subset U^{\mathcal{R}}_{z, y(j)}(\eta),
\]
and the left side has $d_z + j$ elements.

Consider the effect of replacing $\eta_{z,y(\ell)}$ by $\eta'_{z,y(\ell)}$, which shifts $\eta$ to $\tilde{\eta}$.
For $j \neq \ell$, $U^{\mathcal{R}}_{z, y(j)}(\tilde{\eta})$ and
$U^{\mathcal{R}}_{z, y(j)}(\eta)$ differ by at most one element, and the inequality (\ref{e:increasingcf}) implies
\[
 \delta^{z, y(\ell)}_{z,y(j)}  \leq \frac{1}{d_z + j - 1} - \frac{1}{d_z + j}, \quad j \neq \ell.
\]
It follows that the second term of (\ref{e:differencebreakdown}), namely
\(
\sum_{j \neq \ell} \delta^{z, y(\ell)}_{z,y(j)},
\)
is bounded by a telescoping sum, with an upper bound of $\nicefrac{1}{d_z}$, as we sought to prove.
A similar bound of $\nicefrac{1}{d_w}$ applies to the third term of (\ref{e:differencebreakdown})

\noindent{Conclusion: }We have now bounded all the terms of (\ref{e:differencebreakdown}), to
give an upper bound $c_{z, w}$ of the form $1/m_{z, w} + 1/d_z + 1/d_w$, as claimed.
\end{proof}

\begin{cor}\label{c:totalbde}
The Bounded Differences Inequality (\ref{e:bde}) holds for the function $F(\eta)$ in (\ref{e:totalstrangerreciprocal}),
where the constant $\nu$ takes the form:
\[
\nu:=\frac{1}{4} \sum_{\{z, w\} \in \mathcal{R}} c_{z, w}^2 \leq \frac{9 (n-1)^2}{8 K^2},
\]
provided $K:= \min \{K_x,x \in S\} \geq 2$.
\end{cor}
\begin{proof}
Since $\min \{m_{z, w}, d_z, d_w \} \geq K$, it follows that $c_{z, w}:=1/m_{z, w} + 1/d_z + 1/d_w \leq 3/K$.
The set $\mathcal{R}$ of relegated pairs has no more than $n (n - K - 1)/2$ elements, which is less than $(n-1)^2 / 2$
for $K \geq 2$. Hence the  upper bound on $\nu$.
\end{proof}

\subsection{Second application of the Bounded Differences Inequality (\ref{e:bde})}
Fix $x \in S$, and define $F_x$ with respect to a subset of terms appearing in the sum (\ref{e:totalstrangerreciprocal}):
\begin{equation}\label{e:1strangerreciprocal}
   F_x(\eta) :=\sum_{ y \in \mathcal{R}_x } \frac{1}{|U^{\mathcal{R}}_{x, y}(\eta)|}.
\end{equation}

\begin{cor}\label{c:1bde}
The Bounded Differences Inequality (\ref{e:bde}) holds for the function $F_x(\eta)$ in (\ref{e:1strangerreciprocal}),
where the constant $\nu:=\nu_x$ takes the form:
\[
\nu_x:=\frac{1}{4} \sum_{\{z, w\} \in \mathcal{R}} c_{z, w}^2 \leq \frac{ (n-1)^2}{2 K^2}.
\]
\end{cor}

\begin{proof}
The reasoning of Lemma \ref{l:estimatedifferences} about edges in the graph $(S, \mathcal{R})$
is slightly modified.  If $|\{z, w\} \cap \mathcal{P}_x|\geq 1$, then
$U^{\mathcal{R}}_{x, y}(\eta)$ is unaffected by $\eta_{z, w}$. The other two possibilities are:
\begin{enumerate}
    \item [(a)]
Edge $\{z, w\}$ coincides with $\{x, y\}$ for some $y \in \mathcal{R}_x$.
By switching labels if necessary, we may take $z = x$. 
    \item [(b)]
 Both $z$ and $w$ are in $\mathcal{R}_x$, which means $U^{\mathcal{R}}_{x, z}(\eta)$
 and $U^{\mathcal{R}}_{x, w}(\eta)$ could change by 1, when $\eta_{z, w}$ is altered.
\end{enumerate}
Modify formula (\ref{e:differencebreakdown}) to say $|F_x(\tilde{\eta}) - F_x(\eta)|$ is bounded above by
\begin{equation}\label{e:differencebreakdown2}
   1_{ \{z = x \} } 1_{ \{w \in \mathcal{R}_x \} } \sum_{y \in \mathcal{R}_x} \delta^{x, w}_{x, y} +
   1_{ \{\{z, w\} \subset \mathcal{R}_x \} }\left( \frac{1}{m_{x, z}} + \frac{1}{m_{x, w}} \right).
\end{equation}
The first summand arises from case (a), and the second from case (b).
The denominators $m_{x, z}$ and $m_{x, w}$  are minimum sizes of $U^{\mathcal{R}}_{x, z}(\eta)$ and
$U^{\mathcal{R}}_{x, w}(\eta)$, respectively.
Apply the argument of Lemma \ref{l:estimatedifferences} to the first summand to obtain a bound
\[
c_{z, w} =   1_{ \{z = x \} } 1_{ \{w \in \mathcal{R}_x \} }  
\left( \frac{1}{d_x} + \frac{1}{m_{x, w}} \right) +
  1_{ \{\{z, w\} \subset \mathcal{R}_x \} }\left( \frac{1}{m_{x, z}} + \frac{1}{m_{x, w}} \right).
\]
Recall $m_{x, w} \geq K+2$ and $d_x \geq K$.
The sum of squares $\sum\limits_{\{z, w\} \in \mathcal{R} }c_{z, w}^2$ has an upper bound:
\[
(n-K-1)\frac{4}{K^2} +  \frac{(n-K-1)(n-K-2)}{2}\frac{4}{(K+2)^2}
< \frac{2 (n-1)^2}{K^2}.
\]
Divide by 4 to obtain the bound stated in the Corollary.
\end{proof}

\subsection{Third application of the Bounded Differences Inequality (\ref{e:bde})}
To motivate our third application of (\ref{e:bde}), recall how the relegated cohesion 
(\ref{e:randompromoted}) was obtained by a perturbation of a diagonal term
(\ref{e:randomdiagonal}). 
Fix $\{x, v\} \in \mathcal{P}$, and follow the pattern of (\ref{e:randompromoted}) to
knock out at most $d_v$ of the terms appearing in (\ref{e:1strangerreciprocal}):
\begin{equation}\label{e:2strangerreciprocal}
   F_{x, v}(\eta) :=\sum_{ y \in \mathcal{R}_x } \frac{1}{|U^{\mathcal{R}}_{x, y}(\eta)|} -
    \sum_{\substack{y \in \mathcal{P}_v \cap \mathcal{R}_x\\
     y \prec_v x} }
     \frac{1}{|U^{\mathcal{R}}_{x, y}(\eta)|}.
\end{equation}
The analysis is essentially the same as that of (\ref{e:1strangerreciprocal}),
and for every $v$ the bound $\nu:=\nicefrac{(n-1)^2}{(2 K^2)}$ of Corollary \ref{c:1bde} is applicable to
the Bounded Differences Inequality (\ref{e:bde}) for $F_{x, v}$.

\subsection{Proof of Theorem \ref{t:cohesionconcentration}: Concentration inequality}
\begin{proof}
\noindent \textbf{Cohesion matrix: }
Comparison of (\ref{e:randompromoted}) and (\ref{e:2strangerreciprocal}) shows that
for $\{x, v\} \in \mathcal{P}$,
\[
\p[ \left| C^{\mathcal{R}}_{x, v}(\eta) -\E[C^{\mathcal{R}}_{x, v}(\eta)] \right| \geq \theta ]
\leq 2 \p[ F_{x, v}(\eta) -\E[F_{x, v}(\eta)] \geq \theta (n-1) ].
\]
The factor of two arises from consideration of positive and negative deviations from the
mean, for which the bounds are the same. As we just explained, the bound $\nu = \nicefrac{(n-1)^2}{(2 K^2)}$ 
of Corollary \ref{c:1bde} applies to $F_{x, v}(\eta)$ in the Bounded Differences Inequality (\ref{e:bde}). Take $t:= \theta (n-1)$, so the exponent in  (\ref{e:bde}) becomes
\(
t^2 /(2 \nu) = \theta^2 K^2.
\)
The inequality states:
\[
\p[ \left| C^{\mathcal{R}}_{x, v}(\eta) -\E[C^{\mathcal{R}}_{x, v}(\eta)] \right| \geq \theta ]
\leq 2 e^{-t^2/(2 \nu)} = 2 e^{-\theta^2 K^2}.
\]
which proves the first part of Theorem \ref{t:cohesionconcentration}.

\noindent \textbf{Cohesion threshold: }
By (\ref{e:randomdiagonal}), the formula from Corollary \ref{c:pannldthreshold} for $\tau_{\mathcal{R}}$ satisfies
\[
\tau_{\mathcal{R}}:=\frac{1}{2 n} \sum_x \E[C^{\mathcal{R}}_{x, x}(\eta)] =
\frac{1}{n (n-1)} \sum_{ \{x, y\} \in \mathcal{R} } 
\E\left(\frac{1}{|U^{\mathcal{R}}_{x, y}(\eta)|} \right) = \frac{\E[ F(\eta)]}{n (n-1)},
\]
where the last identity follows from (\ref{e:1strangerreciprocal}). A deviation of
$\nicefrac{\theta}{n}$ for $\tau$ is a deviation of $(n-1)\theta$ for $F(\eta)$:
\[  
\p[ | \sum_x C^{\mathcal{R}}_{x, x}(\eta) / (2n) - \tau_{\mathcal{R}} | \geq \frac{\theta}{n} ]
\leq \p[ |F(\eta) -\E[F(\eta)]| \geq (n-1)\theta ].
\]
For $\nu:=9 (n-1)^2 /(8 K^2)$ as in Corollary \ref{c:totalbde}, and for
$t:=(n-1)\theta$, we obtain
\(
t^2 /(2 \nu) = 4 \theta^2 K^2 / 9 = (2 \theta K /3)^2.
\)
The Bounded Differences Inequality (\ref{e:bde}) implies
$\sum_x C^{\mathcal{R}}_{x, x}(\eta) / (2n)$ is concentrated around $\tau_{\mathcal{R}}$:
\[
\p[ | \sum_x C^{\mathcal{R}}_{x, x}(\eta) / (2n) - \tau_{\mathcal{R}} | \geq \frac{\theta}{n} ] \leq
2 \p[ F(\eta) -\E[F(\eta)] \geq (n-1)\theta ] \leq 2e^{-(2 \theta K /3)^2}.
\]
This concludes the proof of Theorem \ref{t:cohesionconcentration}.
\end{proof}

\section{Conclusions and future work}\label{s:prank2xy}

Our \texttt{PaNNLD} approximation to the \texttt{PaLD} 
non-parametric unsupervised learning algorithm is based 
on the \textit{impartiality} of strangers, rather than their \textit{kindness} as in
the opening quote: treating triplet comparisons outside the
$K$-nearest neighbor digraph as random, 
Theorem \ref{t:cohesionconcentration} bounds approximation
error in terms of $K$. Algorithm \ref{a:promocohesion}
illustrates efficient computation of the causal part of
the cohesion matrix, and Section \ref{s:averages} provides
correction terms derived from the local structure of the 
$K$-nearest neighbor digraph. The ingredients were summarized in 
Table \ref{t:pannld-summary}.
The incomplete \texttt{prank2xy} Java implementation \cite{darp}
of the \texttt{PaNNLD} 
clustering and embedding algorithm will be
described in a separate report.

\appendix

\section{Mean reciprocal of conflict focus size in the limit}

\begin{lem}[RECIPROCAL MOMENT APPROXIMATIONS]\label{l:arcsin}
Suppose $U$ is a Uniform$(0,1)$ random variable, and conditional on $U = t$,
$Y^t_{n-m}$ is a Binomial$(n-m, 1 - t^2)$ random variable, for $2 \leq m < n$.
Suppose there exists $c > 1$ such that
\[
n \to \infty, \quad m \to \infty, \quad \frac{m}{n} \to 1 - \frac{1}{c}.
\]
The mean of the bounded random variable $(n-m)/(m + Y^U_{n-m})$ has the following limit:
\[
\lim_{n, m \to \infty} \E \left( \frac{n-m}{m + Y^U_{n-m}} \right)  = 
\int_0^1 \frac{1}{c - t^2} dt = 
\frac{\coth^{-1}(\sqrt{c})}{\sqrt{c}},
\]
As for the variance,
\[
\lim_{n, m \to \infty} \text{Var} \left( \frac{n-m}{m + Y^U_{n-m}} \right) = 
\int_0^1 \frac{1}{(c - u^2)^2} du - \left( \int_0^1 \frac{1}{c - t^2} dt \right)^2.
\]
\end{lem}
\begin{proof}
Set up a probability space on which are defined $U$, and an infinite collection
$(B_i)$ of Bernoulli$(1 - U^2)$ random variables. Take $Y^U_{n-m}$ to be $\sum_{m < i \leq n} B_i$,
so $(Y^U_{n-m})$ have the desired distribution.

Condition on the event $U = t$. 
The strong law of large numbers implies that the bounded random variables $Y^t_{n-m} / (n-m)$
converge almost surely to $1 - t^2$. 
This implies unconditional almost sure convergence
\[
\frac{n-m}{m + Y^U_{n-m}} = \frac{1}{m/(n-m) + Y^U_{n-m}/(n-m)}
\to \frac{1}{c - 1 + (1 - U^2)}, \quad m, n \to \infty.
\]
By Lebesgue's Bounded Convergence Theorem, and symbolic integration  \cite{wol},
\[
\E \left( \frac{n-m}{m + Y^U_{n-m}} \right) \to \E \left( \frac{1}{c - U^2} \right)
= \int_0^1 \frac{1}{c - t^2} dt = \frac{\coth^{-1}(\sqrt{c})}{\sqrt{c}},
\]
where $\coth$ is the hyperbolic cotangent, and $c:=\lim{\frac{n}{n-m}} > 1$.
A similar calculation for the second moment leads to the formula:
\[
\E \left( (c - U^2)^{-2} \right) = \int_0^1 \frac{1}{(c - t^2)^2} dt,
\]
and hence the limiting variance.
\end{proof}


\end{document}